\documentclass[a4paper,UKenglish,cleveref, autoref, thm-restate, numberwithinsect]{lipics-v2021} 

\usepackage[full]{optional} 

\opt{sub}{\usepackage[appendix=append]{apxproof}}
\opt{full}{\usepackage[appendix=inline]{apxproof}}
\opt{final}{\usepackage[appendix=append]{apxproof}} 

\opt{full}{\nolinenumbers}

\usepackage{cite}
\usepackage{graphicx} 

\usepackage[
    textsize=scriptsize,
    disable
]{todonotes} 

\setuptodonotes{fancyline}
\usepackage[ruled]{algorithm2e}
\opt{sub,final}{\usepackage[colorlinks=true,citecolor=blue,linkcolor=blue]{hyperref}}

\usepackage{amsmath}
\usepackage{amsthm}
\usepackage{cleveref}
\usepackage{cancel}
\usepackage[font={small}]{caption}

\opt{sub,final}{\newtheorem{observation}[theorem]{Observation}}

\usepackage{ amssymb }

\def\longrightharpoonup{\relbar\joinrel\rightharpoonup}
\def\longleftharpoondown{\leftharpoondown\joinrel\relbar}
\def\longrightleftharpoons{\mathop{\vcenter{\hbox{\ooalign{\raise1pt\hbox{$\longrightharpoonup\joinrel$}\crcr\lower1pt\hbox{$\longleftharpoondown\joinrel$}}}}}}
\def\rxn{\mathop{\rightarrow}\limits}  

\def\revrxn{\mathop{\rightleftharpoons}\limits}

\renewcommand{\vec}[1]{\mathbf{#1}}
\newcommand{\vx}{\vec{x}}

\newcommand{\vq}{\vec{q}}

\newcommand{\vp}{\vec{p}}

\newcommand{\R}{\mathbb{R}}

\newcommand{\Z}{\mathbb{Z}}

\def\longrightharpoonup{\relbar\joinrel\rightharpoonup}
\def\longleftharpoondown{\leftharpoondown\joinrel\relbar}
\def\longrightleftharpoons{\mathop{\vcenter{\hbox{\ooalign{\raise1pt\hbox{$\longrightharpoonup\joinrel$}\crcr\lower1pt\hbox{$\longleftharpoondown\joinrel$}}}}}}
\def\rxn{\mathop{\rightarrow}\limits}  

\def\revrxn{\mathop{\rightleftharpoons}\limits}

\newcommand{\tp}[1]{{#1^\top}}
\newcommand{\bt}[1]{{#1^\bot}}
\newcommand{\xit}{\tp{x_i}}
\newcommand{\xib}{\bt{x_i}}
\newcommand{\Xit}{\tp{X_i}}
\newcommand{\Xib}{\bt{X_i}}

\newcommand{\calA}{\mathcal{A}}
\newcommand{\calR}{\mathcal{R}}
\newcommand{\calF}{\mathcal{F}}

\newcommand{\calX}{\mathcal{X}}

\usepackage[normalem]{ulem} 
\definecolor{burntorange}{rgb}{0.75, 0.34, 0} 

\usepackage{url}

\newtheoremrep{thm}[theorem]{Theorem}

\opt{full}{\author{David {Doty}}
{University of California--Davis, Davis, CA, USA \and \url{https://web.cs.ucdavis.edu/~doty/}}
{doty@ucdavis.edu}
{https://orcid.org/0000-0002-3922-172X}
{NSF awards 2211793, 1900931 and DoE award DE-SC0024467.}

\author{Mina {Latifi}}{University of California--Davis, Davis, CA, USA \and \url{https://www.linkedin.com/in/mina-latifi/}}{milatifi@ucdavis.edu}
{https://orcid.org/0009-0002-0116-0519}
{NSF awards 2211793, 1900931 and DoE award DE-SC0024467.}

\author{David {Soloveichik}}{The University of Texas at Austin, Austin, TX, USA \and \url{https://www.solo-group.link/}}{david.soloveichik@utexas.edu}
{https://orcid.org/0000-0002-2585-4120}
{NSF awards 2200290, SemiSynBio III: GOALI award, DoE award DE-SC0024467, Schmidt Sciences Polymath Award.}

\authorrunning{D. Doty and M. Latifi and  D. Soloveichik} 



\supplement{
    \emph{Python package:}
    \url{https://pypi.org/project/ode2tn/}
    \\
    \emph{source code repository:}
    \url{https://github.com/UC-Davis-molecular-computing/ode2tn/}
    \\
    \url{https://github.com/UC-Davis-molecular-computing/ode2tn/blob/main/notebook.ipynb}
} 

\acknowledgements{We thank Ophelia S Venturelli for our earlier collaboration on related approaches to analog computation with transcriptional networks, which provided important context for the present work.}}

\begin{document}

\title{Analog computation with transcriptional networks}

\opt{sub,final}{
\author{David Doty\inst{1}\orcidID{0000-0002-3922-172X} \and Mina Latifi\inst{2}\orcidID{0009-0002-0116-0519}
\and
David Soloveichik\inst{3}\orcidID{0000-0002-2585-4120}}
\authorrunning{D. Doty and M. Latifi and  D. Soloveichik} 

\institute{University of California Davis
\email{doty@ucdavis.edu}
\url{https://web.cs.ucdavis.edu/~doty/}
\and University of California Davis \email{milatifi@ucdavis.edu}
\url{https://www.linkedin.com/in/mina-latifi/}
\and The University of Texas at Austin
\email{david.soloveichik@utexas.edu}
\url{https://www.solo-group.link/}
}}

\maketitle

\begin{abstract}
Transcriptional networks represent one of the most extensively studied types of systems in synthetic biology. 
Although the completeness of transcriptional networks for digital 
logic is
well-established,
\emph{analog} computation plays a crucial role in biological systems and offers significant potential for synthetic biology applications.
While transcriptional circuits typically rely on cooperativity and highly nonlinear behavior of transcription factors to regulate protein \emph{production}, they are often modeled with simple linear \emph{degradation} terms. 
In contrast, general analog dynamics require both positive and negative nonlinear terms, seemingly necessitating control over not just transcriptional (i.e., production) regulation but also the degradation rates of transcription factors. 

Surprisingly, we prove that controlling transcription factor production (i.e., transcription rate) without explicitly controlling degradation is mathematically complete for analog computation, achieving equivalent capabilities to systems where both production and degradation are programmable. 
We demonstrate our approach on several examples including oscillatory and chaotic dynamics, analog sorting, memory, PID controller, and analog extremum seeking.  
Our result provides a systematic methodology for engineering novel analog dynamics using synthetic transcriptional networks without the added complexity of degradation control and informs our understanding of the capabilities of natural transcriptional circuits. 

We provide a compiler,
in the form of a Python package that can take any system of polynomial ODEs and convert it to an equivalent transcriptional network implementing the system \emph{exactly}, under appropriate conditions.
\end{abstract}


\section{Introduction}
\label{sec:intro}

A \emph{transcription factor} $X$ is a protein that regulates the transcription (DNA $\to$ RNA) of a gene coding for a protein $Y$, either increasing the rate of production of $Y$ (activation) or decreasing it (repression).
$Y$ could itself be another transcription factor.
A \emph{transcriptional network} is a set of transcription factors that regulate each other in this way.

Complex transcriptional networks have been extensively studied in synthetic biology to implement specific mathematical functions and behaviors, such as oscillation~\cite{elowitz2000synthetic}, bistability~\cite{gardner2000construction},
Boolean logic~\cite{nielsen2016genetic}, and analog function computation~\cite{daniel2013synthetic,sarpeshkar2014analog}.
Many of these studies use transcription factors characterized by first-order (linear) degradation dynamics.
Linear degradation provides an accurate model for studying promoter dynamics in bacteria, particularly in the context of transcriptional rate variability due to growth conditions and extrinsic factors
\cite{rudge2016characterization}.

Various techniques can potentially be employed to construct transcriptional networks with arbitrary complexity and wiring.
One such approach is CRISPR interference (CRISPRi), a gene repression method that silences specific genes without altering DNA sequences~\cite{qi2013repurposing}.
CRISPRi has been used to design synthetic gene circuits, including logic circuits, bistable networks (toggle switches), stripe pattern formation using incoherent feed-forward loops (IFFL), and oscillators \cite{santos2020multistable,gander2017digital}.
In addition to CRISPRi, CRISPR activation (CRISPRa) methods have been developed to activate genes by attaching transcriptional activator domains to dCas9~\cite{chavez2015highly,konermann2015genome}. 
CRISPRa has been used in the design of synthetic systems, such as signal amplification~\cite{tanenbaum2014protein}. 
Recent studies have shown that dual-mode CRISPRa/i architectures can simultaneously enable gene activation and repression within the same system~\cite{moon2025dual}.
Gene regulatory networks can also be implemented in controlled \emph{in vitro} transcriptional systems, such as the synthetic oscillators demonstrated in~\cite{kim2011synthetic}, where interacting DNA templates and transcriptional reactions produce dynamic regulatory behavior.

These studies suggest that transcriptional networks are an expressive and reliable target for implementing analog computation and other sophisticated dynamical behaviors in cells, i.e., a useful \emph{in vivo} programming language.
However, the theoretical limits to their power are poorly understood:
\emph{How} expressive are these networks?
What class of analog computations are they able to achieve?
To answer these questions,
we must formalize a precise model of transcriptional networks.
There are various related approaches to this~\cite{dejong2002modeling}, typically using Hill functions~\cite{alon2019introduction} to control the production rate (positive terms of $\frac{dx}{dt}$) of a transcription factor $X$ (with concentration $x$) and having a single negative term $-\gamma x$, for some constant $\gamma > 0$,
i.e., linear degradation.

The latter constraint follows from the idea that, although mechanisms exist to regulate the \emph{production} of $X$ with other transcription factors, \emph{decay} typically occurs through two mechanisms:
degradation by proteases,
and cell division that increases volume, effectively decreasing concentrations.
When proteases act non-specifically\footnote{
    \Cref{sec:conclusion} discusses this assumption in more detail.
}
and/or the latter mechanism is dominant---the regime studied in this paper---the degradation constant $\gamma$ is the same for all transcription factors.

In contrast to this uniform linear decay, most of the interesting nonlinear behavior in typical transcriptional networks happens in the production rates.
We model these as any \emph{nonnegative Laurent polynomial}, a generalization of multivariate polynomials to allow negative integer exponents,
for instance $x^2 y^{-3} + 4 z  w^{-1} + 5$,
where negative-exponent factors are repressors and positive-exponent factors are activators.
(While this formal class unrealistically allows arbitrarily many activators and repressors per promoter, \Cref{thm:simplification} shows that the construction can be reduced to monomials of the form $\alpha \frac{a_1 a_2}{r}$, corresponding to at most two activators and one repressor, without loss of computational power.)
These positive Laurent-polynomial production rates relate to broader power-law kinetic formalisms that allow negative exponents (often without restricting them to integers), including S-system models in Biochemical Systems Theory~\cite{vilela2008parameter} and generalized-mass-action/power-law systems in CRNT~\cite{arceo2015chemical,boros2018center,muller2016sign}.
We justify in \Cref{sec:justification-laurent} our choice of production rates as approximable by the more standard Hill functions.

Polynomial ordinary differential equations (ODEs) are a noteworthy class of analog computational models,
which have received far more theoretical attention than transcriptional networks.
As shown by Claude Shannon~\cite{shannon1941mathematical},
polynomial ODEs are equivalent to the GPAC (general purpose analog computer) model that Shannon defined to model the capabilities of the differential analyzer machine invented by Vannevar Bush~\cite{bush1931differential} to automate numerical solution of differential equations.
It is known~\cite[Theorem XI]{shannon1941mathematical}
(see also corrected proof in~\cite[Footnote 12]{pour1974abstract})
that ODEs defined by 
non-hypertranscendental functions\footnote{
    Non-hypertranscendental  functions are those that are solutions of algebraic differential equations. This includes all algebraic functions such as polynomials, Laurent polynomials, and $\sqrt{x}$, as well as some transcendental functions such as exponential, logarithm, trigonometric, and hyperbolic functions,
    but excluding, for example, the gamma function $\Gamma$ generalizing factorial (for positive integer $n$, $\Gamma(n) = (n-1)!$) to complex inputs.
}
can be converted into an equivalent set of polynomial ODEs (possibly over a larger set of variables), such that the variables in the polynomial ODEs corresponding to the original variables have the same trajectories.\footnote{
    For example, the non-polynomial ODE 
    $x' = \sqrt{x}$
    is equivalent to the polynomial ODEs 
    $x' = y$ and $y' = 1/2$,
    letting $y = \sqrt{x}$.
    Evidently $x' = y$,
    and by the chain rule,
    $y' 
    = \frac{d \sqrt{x}}{dt}
    = \frac{d \sqrt{x}}{dx} \cdot \frac{dx}{dt}
    = \frac{1}{2 \sqrt{x}} \cdot y
    = \frac{1}{2 \sqrt{x}} \cdot \sqrt{x}
    = 1/2$
    as chosen.
}
Furthermore, polynomial ODEs have recently been discovered to have essentially maximal \emph{digital} computational power, being able to simulate arbitrary Turing machines~\cite{bournez2017odes}.
Being so expressive, polynomial ODEs are an attractive target for implementation by other analog models of computation.

Compared to polynomial ODEs, 
the primary limitation of transcriptional networks is that the ODE for a transcription factor $X$ has exactly one negative term and it is of the linear form $-\gamma x$, where $\gamma$ is the same for all transcription factors.
In contrast, polynomial ODEs can have arbitrarily complex negative terms.
It thus appears difficult to simulate arbitrary polynomial ODEs with a system that has such a strong limitation on its negative terms.

Surprisingly, our main result, \Cref{thm:ODEtoTN} (\Cref{sec:results}), shows that ``almost''\footnote{
    The condition we require the polynomial ODEs to satisfy, other than staying nonnegative, is that any variable that starts at 0 or converges to 0 must have a ``Hungarian form'' ODE, 
    which roughly means ``is the ODE of some chemical reaction network''; 
    see \Cref{def:hungarian}.
} any system of polynomial ODEs whose variables stay nonnegative can be ``ratio-implemented'' by a transcriptional network.
This means that for each variable $x$ in the original ODEs, there is a pair of transcription factors $\tp{X},\bt{X}$ (``$X$-top'' and ``$X$-bottom''), such that for all times $t$,
$\frac{\tp{x}(t)}{\bt{x}(t)} = x(t)$, where $\tp{x}(t) , \bt{x}(t)$ represent the concentrations of $\tp{X},\bt{X}$ at time $t$ respectively.
Intuitively, the reason that this ratio representation of values helps is that, if all transcription factors decay at rate $\gamma$ for the same amount of time, this preserves all ratios between them.

The nonnegativity requirement is not a fundamental limitation. One can use a \emph{dual-rail} representation (c.f.,~\cite{oishi2011biomolecular,cardelli2020electric}),
in which each variable $x$ is encoded as the difference of two nonnegative quantities $x^+ - x^-$, allowing arbitrary real values to be represented while maintaining nonnegativity of all variables. This also avoids issues with variables approaching zero in non-Hungarian form, since both $x^+$ and $x^-$ remain positive even when their difference converges to zero.
In some of our examples such as the sine-cosine oscillator, we use a simpler ad-hoc trick of ``shifting values up'' to maintain nonnegativity; see~\Cref{sec:examples}.
(This is not a general-purpose technique since it requires a lower bound on the negative values to know how much to shift; for example since sine and cosine are at least -1, we shift up by 2, and could shift by any $\alpha > 1$, to maintain positivity.)

A na\"{i}ve implementation of the ratio-representation idea has the property that even if the variables of the original ODEs stay bounded 
(both above by some finite upper bound $u$, 
and also bounded away from 0 by some lower bound $\ell > 0$),
the transcription factors could either diverge to $\infty$ or converge to 0,
despite the ratios $\frac{\tp{x}}{\bt{x}}$ staying in the interval $(\ell,u)$.
(This simpler construction is shown in the proof of \Cref{thm:ODEtoTN} as an intuitive warmup to the full construction.)
Our construction (see \Cref{alg:construction}) overcomes this obstacle, maintaining the useful property that if the original variables were bounded in some interval $(\ell, u)$,
then the new transcription factor variables are bounded in some (possibly larger) interval $(\ell', u')$.
The precise condition is somewhat complex and captured in \Cref{def:implementation,def:bounded-positive,def:bounded-implementation}.

Next (\Cref{thm:simplification}, \Cref{sec:simplification}), we show how to trade off ``complexity'' in this construction.
It is reasonable to worry that using arbitrarily many activators and repressors on a single promoter site pushes the regulation model too far.
We show how to modify the construction of \Cref{alg:construction} so that each Laurent monomial term is at most of the form $\alpha \frac{a_1 a_2}{r}$,
i.e., each promoter site is regulated by at most two activators ($a_1,a_2$)
and one repressor ($r$).
The tradeoff is that additional transcription factors must be introduced in addition to the pair $\xit,\xib$ for each original polynomial ODE variable $x$.
This is done by first converting the polynomial ODEs to have extra variables $z_{i,j}$ corresponding to the $j$'th monomial on the right-hand side of $\dot x_i$.
Each such new variable then would get its own pair of transcription factors $\tp{z_{i,j}}, \bt{z_{i,j}}.$
One additional constraint for this construction to work is that all variables must start strictly positive,
though we are unsure if this is truly necessary or simply helps our proof,
whereas \Cref{thm:ODEtoTN} allows ``Hungarian'' variables (\Cref{def:hungarian}) to start at 0 before becoming positive.

Finally, we provide a compiler that implements the main construction of \Cref{thm:ODEtoTN}, 
the Python package \texttt{ode2tn}~\cite{ode2tn}.
All examples in \Cref{sec:examples} were simulated using \texttt{ode2tn}.

We note prior unpublished work studying the computational power of transcriptional networks for analog computation with a different approximation of Hill-function kinetics than in the current paper~\cite{venturellicompleteness}.
\todo{DS: Moved to the end of intro}

\section{Preliminaries}
\label{sec:prelim}


Throughout this paper, for a real variable $x$, a function of time $t$, we sometimes write $x'$ to denote $\frac{dx}{dt}$, the derivative of $x$ with respect to $t$.

We recall the definition of a Laurent polynomial, which generalizes the idea of a (multivariate) polynomial to allow negative integer exponents.

\begin{definition}
[Laurent polynomial]
\label{def:laurent}
A \emph{Laurent polynomial} is of the form: 
$p(x_1,\ldots,x_n) =  \sum_{k=1}^\ell c_k \prod_{i=1}^n x_i^{e_{k,i}}$, 
where each $c_k \in \R$
and $e_{k,i} \in \Z$ (not necessarily positive).
If each $c_k \ge 0$, we say $p$ is \emph{nonnegative}.
If $\ell = 1$, we say $p$ is a \emph{Laurent monomial}.
\end{definition}

For example $\frac{x^2}{y^3} + 4 z^5 + \frac{5}{x} - \frac{6}{w^2}$ is a Laurent polynomial consisting of four Laurent monomials.
It is not nonnegative because the coefficient of the last monomial is $-6$.
Note that a (standard) polynomial is a Laurent polynomial with all $e_{k,i} \geq 0$. 
We similarly say a polynomial is \emph{nonnegative} if all of its coefficients $c_k \ge 0$.

We now define a class of ODEs that are of interest because they correspond precisely to the class of ODEs that represent the dynamical systems known as chemical reaction networks (CRNs)~\cite{cardelli2020electric}.
We follow the convention throughout this paper that for a chemical species $X$ (including transcription factors discussed below) written in uppercase, its concentration is represented by the same lowercase letter $x$.
CRNs consist of reactions among abstract species such as $A+2B \rxn^{k} 3C$.
Each reaction's rate is the product of its reactant concentrations and the \emph{rate constant} written above the arrow, and 
each species $X$ has a positive (respectively, negative) derivative term for each reaction net producing (resp., consuming) $X$.
For example, the reaction above would correspond under the standard mass-action model of kinetics to the ODEs
$
a' = 
-kab^2,
b' = 
-2kab^2,
c' = 
3kab^2.
$

A negative term for $x'$ means in the reaction that contributes to that term, $X$ is net consumed, 
thus is a reactant, 
implying $x$ must appear as a factor in the negative term, 
motivating the following definition, originating from~\cite{hars1981inverse}, with terminology taken from~\cite{cardelli2020electric}.

\begin{definition}
\label{def:hungarian}
A system of polynomial ODEs over variables $x_1,\dots,x_n$ is in \emph{Hungarian form} for a particular variable $x_i$ if
there are nonnegative polynomials $p_i^+$ and $q_i^-$  such that
\[
x_i' = 
p_i^+(x_1,\dots,x_n) - x_i \cdot q_i^-(x_1,\dots,x_n).
\]
\end{definition}

In other words $q_i^-$ consists of the negative terms of $x_i'$, divided by $x_i$.
We say that a system of polynomial ODEs is in \emph{Hungarian form} (without reference to a variable) if it is in Hungarian form for all its variables. 
In fact, a system of ODEs is in Hungarian form if and only if it is the set of ODEs describing the kinetics of some CRN~\cite{hars1981inverse}.
It is also known that a set of polynomial ODEs is in Hungarian form if and only if, for all nonnegative initial conditions, the variables stay nonnegative.

\subsection{Transcriptional networks}

Intuitively, a transcriptional network is a system of chemical species, called \emph{transcription factors},
where each factor $X$ is produced at a rate depending on other factors, and decays at rate $\gamma x$,
where $\gamma>0$ is the same for all transcription factors.
The factors regulating $X$'s production rate $p$ are either \emph{activators}, increasing $p$, or \emph{repressors}, decreasing $p$ (\cref{fig:TN-example}).
We allow the production rate to be an arbitrary Laurent polynomial whose variables are transcription factor concentrations; this choice is justified in \Cref{sec:justification-laurent}.

\begin{definition}
\label{defn:transcriptional-network}
A \emph{transcriptional network} is a system of ordinary differential equations (ODEs) over variables $x_1,\dots,x_n$ representing \emph{transcription factor} concentrations
\begin{eqnarray*}
    x_1' &=& p_1(x_1,x_2, \dots x_n) - \gamma x_1 
    \\
    x_2' &=& p_2(x_1,x_2, \dots x_n) - \gamma x_2
    \\&\vdots&\\
    x_n' &=& p_n(x_1,x_2, \dots x_n) - \gamma x_n,
\end{eqnarray*}
where each $p_i$ is a nonnegative Laurent polynomial, and $\gamma > 0$ is called the \emph{decay constant}.
\end{definition}

Positive exponent factors such as $a$ or $a^2$ in Laurent polynomials for $x'$ imply that $A$ is an activator for $X$;
negative exponent factors such as $\frac{1}{r}$ imply that $R$ is a repressor for $X$.

\begin{figure}[!ht]
\centering
\includegraphics[width=0.9\linewidth]{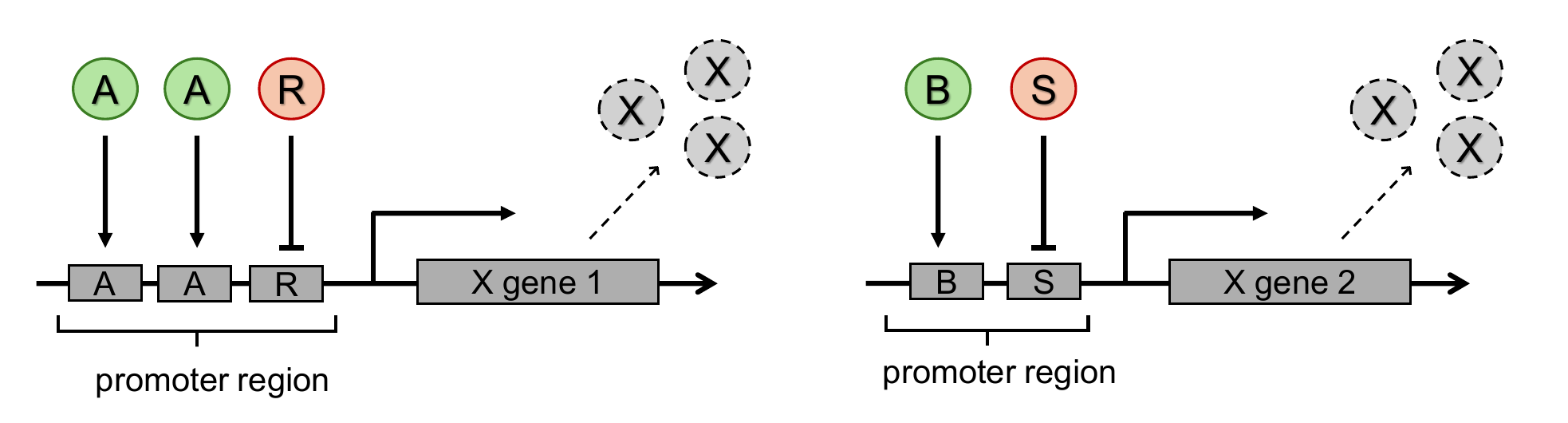}
\caption{
    The transcription factors, activator $A$ and repressor $R$, regulate $X$ gene 1 (shown on the left) controlling the production rate of protein $X$ 
    (assumed to be produced instantly on translation,
    though in reality only an mRNA strand would be produced by transcription, and subsequent ribosomal translation of the mRNA into a protein would be necessary to produce $X$). Similarly, activator $B$ and repressor $S$, regulate $X$ gene 2 (shown on the right).
    Each transcription factor has a specific binding site on the promoter, allowing it to bind independently of the other transcription factors.
    In cells, $X$ may or may not be a transcription factor itself;
    in this paper we assume that all genes encode transcription factors.
    Based on \Cref{defn:transcriptional-network}, the ODE for transcription factor $X$ is $x' = a^2/r + b/s - \gamma x$.
    (Uppercase letters denote transcription factors, and lowercase letters denote their concentrations.)
}  
\label{fig:TN-example} 
\end{figure}    

\begin{toappendix}
    
\subsection{Justification of Laurent polynomials as class of production rates}
\label{sec:justification-laurent}

When the binding of the activator $A$ is weak yet the binding of the repressor $R$ is strong, 
it is common to approximate the transcription rate as proportional to the ratio $a/r$; 
see, for example, ref.~\cite{goentoro2009incoherent} and \cite[Section 10.4]{alon2019introduction}.
We now present a self-contained argument, generalizing this approximation to multiple activators and repressors, and several copies of the same transcription factor gene, resulting in the Laurent polynomial form of transcription rate.

Transcription factors are proteins that bind to a specific DNA sequence,
known as a \emph{promoter region}, 
to regulate the transcription of a nearby gene~\cite{alon2019introduction}, ultimately influencing the amount of protein the gene encodes.
These transcription factors can function as activators, increasing the rate of transcription (production), or as repressors, decreasing it.
Each transcription factor undergoes decay due to a combination of dilution (due to cell division) and active degradation. 
For the purposes of this work, we assume that the rate of degradation is the same for every transcription factor (i.e., dilution or degradation by non-specific enzymes),
implying a \emph{uniform linear decay} rate $\gamma$ across all transcription factors.

A transcriptional network describes the interactions between transcription factors and their target genes, where the transcription factors  regulate genes that  themselves encode \emph{other transcription factors}.
Thus all transcription factors are produced at a rate determined entirely by the current concentrations of transcription factors in the network.
An activator $A$ of a transcription factor $X$ increases $X$'s production rate, 
and a repressor $R$ decreases $X$'s production rate (\cref{fig:TN-example}).

The mechanism of regulation of $X$ is by influencing the rate at which RNA polymerase can bind to the promoter region upstream of the gene for $X$. 
If the RNA polymerase successfully binds and transcribes, i.e. creates an mRNA strand representing $X$, then more $X$ will be produced (after the mRNA is translated into the protein $X$ by a ribosome, assumed in this model to happen negligibly soon after transcription.)
Transcription factors affect production rate by influencing the probability that an RNA polymerase successfully binds and transcribes.
There is some chance that a polymerase will detach on its own before successfully beginning transcription;
activators help the polymerase stick around for longer and increase the chances of successful transcription.
Repressors simply block the polymerase by getting in the way.

These influences of a transcription factor on its target gene can be described using a Hill function, which is derived from the probability of the transcription factor binding to the gene's promoter region.
The Hill functions for an activator $A$ and a repressor $R$, with respective concentrations $a$ and $r$, are defined as follows. Here, $K_A,K_R \in \R^+$ respectively represent activation and repression coefficients, $\alpha_a,\alpha_r \in \R^+$ denote maximal promoter activity:
\begin{equation}\label{eq:Hill_act}
   H(a) = \alpha_a \frac{a}{K_A + a} 
\end{equation}

\begin{equation}\label{eq:Hill_rep}
   H(r) = \alpha_r \frac{K_R}{K_R + r} 
\end{equation}

The term $\frac{a}{K_A + a}$ denotes the probability that the activator $A$ is bound to the promoter.
Similarly, the term $\frac{K_R}{K_R + r}$ denotes the probability that the repressor $R$ is not bound.
This expression arises from the equivalence $1 - \frac{r}{K_R + r} = \frac{K_R}{K_R + r}$.

The fully general Hill function also has a ``cooperativity'' constant $c$, appearing as an exponent in the terms above like $K_A^c$ and $a^c$.
In this paper we assume the non-cooperative case of $c=1.$
In other words, each transcription factor binds independently to its own specific binding site on the promoter, without being influenced by the presence of other transcription factors binding to the same gene. 

We also assume that a transcription factor $X$ can potentially be produced by multiple copies of its gene,
each of which, due to different promoter regions upstream of each copy, is regulated by different transcription factors.\footnote{
    This is ultimately our justification to go from single-term Laurent monomials to multi-term Laurent polynomials; each gene copy contributes additively to the production rate.
    Most of this section justifies the use of Laurent monomials, in place of products of Hill functions,
    for describing the regulation of a single copy.
}
Suppose that a transcription factor $X$ is encoded in $m$ different copies $\{g_1,g_2,...,g_m\}$ of the gene,
each regulated by multiple transcription factors,
with the $i$'th copy having activators $\calA_i \subseteq \calF$ and repressors $\calR_i \subseteq \calF$.
Then the production rate of $X$ can be expressed as:
\[
\sum_{i=1}^m \alpha_i
\prod_{A \in \calA_i} \frac{a}{K_A + a}
\prod_{R \in \calR_i} \frac{K_R}{K_R + r}
\]

We can simplify this summation if we factor out $K_R$'s from the numerator and include them in the constant ($\alpha_i' = \alpha_i \prod_{R \in \calR_i} K_R$) to get


\[
\sum_{i=1}^m \alpha_i' \prod_{A \in \calA_i} \frac{a}{K_A + a} 
\prod_{R \in \calR_i} \frac{1}{K_R + r}
\]

We note that if $K_R$ is assumed to be very small and factored out for each repressor, then the coefficient varies more significantly as the number of repressors increases. To show that a wide range of coefficients in a Laurent monomial can be achieved, we introduce a modified construction in \cref{sec:simplification} that restricts each promoter to have only a fixed number of repressors and activators.




\begin{remark}
\label{rem:proprtional}
    Suppose $\calA = \{ A_1,A_2, \dots, A_m\}$ and $\calR = \{R_1,R_2,\dots,R_l\}$ respectively represent the sets of activators and repressors associated with a specific promoter.
    If the activation coefficients  ($K_{A_i}$) are large enough compared to the activator's concentration,
    and the repression coefficients ($K_{R_j}$) are small enough compared to the repressor's concentration,
    then the transcriptional rate of this promoter is proportional to the ratio of the product of the activator concentrations to the product of the repressor concentrations, which is a Laurent monomial.
    \[
        \alpha_i'\prod_{A_i \in \calA} 
        \frac{a_i}{K_{A_i} + a_i}
        \prod_{R_j \in \calR}
        \frac{1}{K_{R_j} + r_j} 
        \approx \alpha_i'' 
        \frac{\prod_{A_i \in \calA} a_i}{\prod_{R_j \in \calR} r_j}. 
    \]
    Note that if a repressor approaches zero, then no finite value of $K_{R_j}$ will satisfy the approximation. 
Thus we need to make sure that repressors stay bounded away from zero (see \Cref{rem:repressors-do-not-approach-zero}).
\end{remark}

In other words,
    in the limit of large $K_{A_i}$ and small $K_{R_j}$, the probability of activator binding decreases while the probability of repressor binding increases, 
and thus the rate of transcription decreases. 
Mathematically, we can assume that the maximal promoter activity $\alpha$ increases to balance the effect,
but physically there is a resulting speed-accuracy tradeoff that is worth further exploration.

\Cref{rem:proprtional} justifies our use of nonnegative Laurent polynomials as the class of production rates available in programming transcriptional networks.
By nonnegative we mean that all the Laurent monomials have nonnegative coefficients, corresponding to the control of the rate of production.

\end{toappendix}

\section{Transcriptional networks can implement nonnegative polynomial ODEs}
\label{sec:results}
\begin{toappendix}
\label{apx:results}
\end{toappendix}

This section demonstrates how transcriptional networks can implement any system of polynomial ordinary differential equations (ODEs) with the caveats given in \cref{thm:ODEtoTN}.
The key concept is that variables $x_i$ in the original system are represented as a ratio of two transcription factor concentrations $x_i = \xit / \xib$. 
This ratio representation seems natural in the context of identical linear decay since it remains unchanged when both the numerator and denominator decrease at the same first-order rate (i.e., $-\gamma \xit$ and $-\gamma \xib$).

However, we must develop technical machinery to address the following questions:
(1) How to control the production rate of $\tp{X_i}$ and $\bt{X_i}$ to ensure that the ratio $x_i = \xit / \xib$ traces the original trajectory?
(2) Note that the ratio representation may correctly represent the value $x_i$, but both the numerator and denominator go to infinity or to zero.
Clearly transcription factors must remain bounded, and, further, we cannot expect the system to be reliable if the ratio is represented with very small values in both the numerator $\xit$ and denominator $\xib$ to represent a much larger $x_i$,
since small values are subject to more environmental noise.\footnote{
    That said, it is realistic to expect that the original ODEs might converge $x_i$ to 0; in which case the transcriptional network should of course converge  $\xit$ to 0 as well.
}
How can we avoid such instability?


A \emph{trajectory of a dynamical system} (e.g., a set of polynomial ODEs)
over a set of variables $x_1, \dots, x_n$ is a function $\rho: \R_{\geq 0} \to \R^n$, where $\rho(t)_{i}$ is the value of $x_i$ at time $t$, also written as $x_i(t)$.
The trajectory may not be defined over all time; then we allow $\rho: [0, t_\mathrm{max}) \to \R^n$ for $t_\mathrm{max} \in \R_{\ge 0} \cup \infty$.\footnote{
    For most systems $t_\mathrm{max} = \infty$.
    However, even natural systems such as CRNs can have pathological behavior:
    for example the reaction $2X \rxn 3X$, starting with $X(0)=1$, has solution $X(t) = 1 / (1-t)$, 
    so that $\lim_{t\to 1} X(t)=\infty$.
    However, this trajectory is defined and finite on the bounded domain $[0,1)$ (i.e., $t_\mathrm{max}=1$),
    and indeed, the construction of \Cref{thm:ODEtoTN} applied to that CRN gives a transcriptional network where $\lim_{t\to 1} \tp{x}(t) / \bt{x}(t) = \infty$, yet $\tp{x}(t)$ and $\bt{x}(t)$ are finite for all $t < 1$.
}

\begin{definition}
\label{def:implementation}
Let $t_\mathrm{max} \in \R_{\ge 0} \cup \infty$ and let $\rho: [0, t_\mathrm{max}) \to \R_{\ge 0}^n$ be a trajectory of a dynamical system defined for all  $t \in [0,t_\mathrm{max})$. 
We say a transcriptional network $\calF$ 
\emph{ratio-implements}
$\rho$ 
if each variable $x_i$ of $\rho$ is mapped to a pair of transcription factors $\Xit,\Xib$ of $\calF$ such that $\calF$ is defined on $[0, t_\mathrm{max})$ and the following holds.\footnote{
    Certainly if $x_i(t)$ is defined for all $t \in [0, t_\mathrm{max})$, so must be $\xit$ and $\xib$ for their ratio to obey $x_i(t) = \xit(t)/\xib(t)$ on that interval, but the definition allows (as in \Cref{thm:simplification}, for instance) to introduce auxiliary variables, and these must be be defined on that interval as well.
}
For each variable $x_i$ of $\rho$, for any $\xit(0) \in \R_{\ge 0}$, $\xib(0) \in \R_{> 0}$ such that $x_i(0) = \xit(0)/\xib(0)$,
for all future times $0 < t < t_\mathrm{max}$,
$x_i(t) = \xit(t) / \xib(t)$.
\end{definition}

Of course, the underlying polynomial ODEs may diverge to $\infty$ or converge to 0 for some variables; if so then certainly the transcriptional network must do the same (e.g., sending $x$ to 0 either by sending $\tp{x}$ to 0 or $\bt{x}$ to $\infty$).
However, if the original variables stay bounded in an open interval $(\ell,h)$ for $0 < \ell < h < \infty$, then we would like the transcription factors to stay bounded in this way as well.

\begin{definition}
\label{def:bounded-positive}
A trajectory $\rho$ is \emph{bounded-above} on a set of variables $X$ if for all $x \in X$, there is $b_\mathrm{max}$ such that for all $t \in [0,t_\mathrm{max})$, $x(t) < b_\mathrm{max}$
(i.e., $x(t)$ does not diverge,
i.e., 
{$\sup_{0 \le t < t_\mathrm{max}}$}
$x(t) < \infty$).    
Similarly, $\rho$ is \emph{bounded-positive} on a set of variables $X$ if for all $x \in X$, 
there is $b_\mathrm{min} > 0$ such that for all $t \in [0,t_\mathrm{max})$, $x(t) > b_\mathrm{min}$.
We say $\rho$ is \emph{bookended} on a set of variables $X$ if it is both bounded-above and bounded-positive on $X$.
\end{definition}
When we omit the set of variables $X$, then bounded-above, bounded-positive, and bookended means with respect to all the variables $X = \{x_1, \dots, x_n\}$.

\begin{definition}
\label{def:bounded-implementation}
Suppose a transcriptional network $\calF$ ratio-implements a trajectory $\rho$.
We say that $\calF$ ratio-implements $\rho$ \emph{stably}
if the trajectory of $\calF$ is bounded-positive on the set of all $x_i^\bot$, and if $\rho$ being bounded above on all variables implies that $\calF$ is bounded above on all variables
(thus bookended on all $\xib$).
\end{definition}

Note that the final part of \Cref{def:bounded-implementation} does not follow from the first and the correctness of the ratio-representation:
it is possible that in $\rho$ some variable $x_i$ goes to 0,
yet possibly $\calF$ ratio-implements this by letting $\xib$ go to $\infty$.
We want to avoid this,
since it is unrealistic to have concentrations approach $\infty$.


Note that if $\rho$ diverges on even a single variable, then possibly any transcription factor concentration $\xit$ or $\xib$ can diverge while maintaining the correct ratio $x_i = \xit / \xib$.
On the other hand if $\rho$ does not diverge, 
then every transcription factor in a stable ratio-implementation remains bounded-above.
Ensuring that $\xib$ is bounded-positive avoids the other issue identified at the beginning of this section: 
the possibility that both $\xit$ and $\xib$ go to zero while maintaining the correct ratio.

We note that to show an implementation of a bounded-above $\rho$ is stable, it suffices to show that the $\xib$ are bookended.

\begin{observation}
\label{obs:bdd-x-xt}
Due to the correctness of ratio-implementation, i.e., maintaining that $x_i(t) = \frac{\xit(t)}{\xib(t)}$ for all $t \in [0,t_\mathrm{max})$, if $x_i^\bot$ is bookended, then 
$x_i$ is bounded-above if and only if $\xit$ is bounded-above.
\end{observation}
In other words, if we interpret diverging as a ``realism'' constraint, then a stable implementation is precisely as realistic as the ODEs being implemented.


One could generalize the definition of ratio-implementation to other representations, including the direct representation where each variable $x$ in the original dynamical system is represented directly by a transcription factor $x$,
or perhaps a mix of the ratio representation for some variables and the direct representation for others. 
For simplicity, we state the definition in terms of the ratio representation studied in this paper.
However, in the ``extremum-seeking'' example in \Cref{sec:extremum}, we use a mixed representation.

Our main theorem says that a transcriptional network can implement arbitrary polynomial ODEs from any initial condition that stays nonnegative,
so long as any variable whose ODE is not in Hungarian form stays bounded away from 0.
(Note that staying bounded away from 0 means starting strictly positive, and also not converging to 0.)
We note that $\gamma$ is a parameter of \Cref{alg:construction}, not an output,
since in general $b_\mathrm{max}$ is uncomputable from the ODEs and initial conditions,
due to the ability of polynomial ODEs to simulate Turing machines~\cite{bournez2017odes}.
\opt{sub,final}{
A full proof of \Cref{thm:ODEtoTN} is given in  \Cref{apx:results}.
}

\todo{DD: We should emphasize that $\gamma$ cannot be computed from $P$, since PODEs are Turing-universal. In the Python package I think we actually have $\gamma$ be an input, but we should emphasize that there's no way in general to compute $\max_i \sup_{t \ge 0} p_i^-(\vec{x}(t))/x_i(t)$, so $\gamma$ is maybe better thought of as an input to the algorithm, with an (uncomputable) constraint it must obey.}

\begin{algorithm}[!ht]
\caption{Construction of a transcriptional network $\calF(\gamma)$ from polynomial ODEs $P = (p_1,\dots,p_n)$, with
$\vec{x}(t) = (x_1(t),\dots,x_n(t) )$ representing the variables at time $t$. Note that the input is just the system of polynomial ODEs, but the constant $\gamma$ below is chosen also based on bounds from a trajectory from specific initial conditions.}\label{alg:construction}

Rewrite each polynomial $p_i$ as 
$x_i' = p_i^+(x_1,x_2,\dots,x_n) - p_i^-(x_1,x_2,\dots,x_n)$ where $p_i^+$ and $p_i^-$ are non-negative polynomials.

For each $x_i$ include two transcription factors $X_i^\top$ and $X_i^\bot$ with initial concentrations $x_i^\top(0) \geq 0$ and $x_i^\bot(0) > 0$ such that: 
$x_i(0)=\frac{x_i^\top(0)}{x_i^\bot(0)}$. 

Construct the transcriptional network as:
\begin{align}
\frac{d \xit}{dt} 
&= \beta \xit/\xib + p_i^{+} \xib - \gamma \xit
\label{eq:new-top}
\\
\label{eq:new-bot}
\frac{d \xib}{dt} 
&= \beta + p_i^{-} \xib^2/\xit - \gamma \xib 
\end{align}
where $\beta > 0$ is any positive constant. $\gamma$ is a parameter of the construction and must obey $\gamma > \max_i \sup_{t \ge 0} p_i^-(\vec{x}(t))/x_i(t)$. 
Equations \eqref{eq:new-top} and \eqref{eq:new-bot} are assumed to be fully simplified; see \Cref{rem:eqs-need-simplification} and \Cref{rem:repressors-do-not-approach-zero}.
\end{algorithm}

\begin{thmrep}
\label{thm:ODEtoTN}
Let $P$ be a system of polynomial ODEs over the set of variables $\calX$ that is in Hungarian form for the subset of variables $\calX_h \subseteq \calX$ (possibly empty set).
Then there is a transcriptional network $\calF(\gamma)$, constructed by Algorithm~\ref{alg:construction}, 
such that for every upper-bound $b_\mathrm{max} > 0$,
and every lower-bound $b_\mathrm{min} > 0$, 
there is $\gamma > 0$, 
such that for every nonnegative trajectory $\rho$ of $P$ that is bounded-above by $b_\mathrm{max}$ on $\calX_h$ and bookended by $b_\mathrm{max}$ and $b_\mathrm{min}$ on $\calX \setminus \calX_h$,
$\calF(\gamma)$ stably ratio-implements $\rho$.
\end{thmrep}

Since we allowed the definition of transcriptional networks to have arbitrary Laurent polynomials (a strict superset of polynomials) as production rates,
we could also have simulated arbitrary systems of ODEs with Laurent polynomials for each derivative.
In other words, $p_i^+$ and $p_i^-$ could be nonnegative Laurent polynomials.\footnote{  
Note that we could handle even more functions on the right-hand side of the ODEs (any non-hypertranscendental functions) by first converting to a system of polynomial ODEs~\cite{pour1974abstract}.
However, the construction of \Cref{alg:construction} works \emph{directly} as shown
(without any initial pre-transformation of the ODEs)
on Laurent polynomials,
hence could result in a simpler final transcriptional network than one obtained by first converting to polynomial ODEs.}
This fact will be used later in \Cref{thm:simplification}.


\begin{proof}
We begin with a simpler ``warmup'' construction that ratio-implements $\rho$, but not stably,
before describing how to modify it to be stable.
By grouping positive and negative terms together, any polynomial $p(\vx)$ can be written as $p^+(\vx) - p^-(\vx)$, where $p^+$ and $p^-$ are polynomials with all nonnegative coefficients.
Thus any polynomial ODE can be expressed as
\[
x_i' = p_i^+(x_1,x_2,\dots,x_n) - p_i^-(x_1,x_2,\dots,x_n).
\]
$x_i$ can be represented as the ratio of two transcription factors $\xit$ and $\xib$:
$
x_i = \frac{\xit}{\xib},
$
where the initial concentrations of $\tp{X_i}$ and $\bt{X_i}$ are any nonnegative values such that the ratio $\xit / \xib$ is the desired initial value of $x_i = \rho(0)_i$. 
For brevity in the equations below, we write $p_i^+$ to denote $p_i^+(x_1,\dots,x_n)$,
and similarly for $p_i^-$.
 


We claim that the following transcriptional network ratio-implements $\rho$,
where $\gamma > 0$ is a constant to be determined later.

\begin{eqnarray}
\label{eq:x_top}
    \xit' &=&  p_i^{+} \xib - \gamma \xit
    \\
\label{eq:x_bot}
    \xib' &=& p_i^{-} \xib/x_i - \gamma \xib
\end{eqnarray}
Note that \eqref{eq:x_bot} can be equivalently written $\xib' = p_i^{-} \xib^2/\xit - \gamma \xib$; here, $\xit$ and $\xib$ are transcription factor concentrations, whereas $x_i$ is the abstract value $\frac{\xit}{\xib}$ but is not itself the concentration of any transcription factor in the system.
(See \Cref{rem:proprtional} for how to introduce such a transcription factor whose concentration is $\frac{\xit}{\xib}$, if that is desired.)
To justify that this transcriptional network ratio-implements $\rho$,
we must show that $\frac{d(\xit/\xib)}{dt} = p_i^+ - p_i^-$.

\begin{align*}
    \frac{d(\xit/\xib)}{dt} 
    &= \frac{\xit'\xib - \xit\xib'}{\xib^2}
    && \text{quotient rule}
    \\
    &= \frac{\xit'}{\xib} - \frac{\xit \xib'}{\xib^2}
    \\
    &= \frac{p_i^{+} \xib - \gamma \xit}{\xib} - \frac{\xit(p_i^{-} \xib/x_i - \gamma \xib)}{\xib^2}
    && \text{substituting by \eqref{eq:x_top} and \eqref{eq:x_bot}}
    \\
    &= p_i^{+} - \frac{\gamma \xit}{\xib} - p_i^{-} + \frac{\gamma \xit}{\xib}
    = p_i^{+} - p_i^{-}.
\end{align*}

The transcriptional network given here, though it satisfies \Cref{def:implementation}, has an important problem: 
It is not a stable ratio-implementation by \Cref{def:bounded-implementation}.
To see this, consider equation \eqref{eq:x_bot}. We can rewrite it as
\[
\frac{d \xib}{dt} = \xib( p_i^{-} /x_i - \gamma).
\]

Note that if $p_i^{-} /x_i - \gamma$ remains sufficiently positive (respectively, negative) for sufficient time, 
then the value $\xib$ over that timescale will exponentially increase (resp., decrease).
Thus,
while the ratio $\xit/\xib$ can remain in a positive bounded interval (i.e., bound by $(\ell,h)$ for $0<\ell<h<\infty$), 
both the numerator and denominator can either increase without bound 
(if $p_i^{-} /x_i - \gamma$ stays positive) or converge to $0$ 
(if $p_i^{-} /x_i - \gamma$ stays negative).
While in some cases it might be possible to tune the dynamics such that the system oscillates between these two regimes, in general this is difficult and the transcription factor concentrations $\xit$ and $\xib$ will tend to both exponentially increase or go to zero.
Thus we choose $\gamma > \max_i \sup_{t \ge 0} p_i^-(\vec{x}(t))/x_i(t).$
For this it suffices to choose $\gamma > C b_\mathrm{max}^k / b_\mathrm{min}$,
where $C$ is the maximum number of monomials in any Laurent polynomial ODE, and $k$ is the maximum degree; thus $C b_\mathrm{max}^k$ upper bounds the value $p_i^-(\vec{x}(t))$ for any $t$.

We now show a modified construction (see Algorithm~\ref{alg:construction}) that also satisfies the definition of stable ratio-implementation (\Cref{def:bounded-implementation}).
Recall that $\rho$ is bounded-above on all $x_i$ variables.
Recall that \Cref{def:bounded-implementation} requires that we show all $\xit,\xib$ transcription factor concentrations are bounded above,
and furthermore, all $\xib$ are also bounded-positive.

To fix this problem, we include a positive 
constant $\beta$ in the production rate of $\xib$ and change the transcriptional network to:

\begin{align*}
\frac{d \xit}{dt} 
&= \beta x_i + p_i^{+} \xib - \gamma \xit
\\
\frac{d \xib}{dt} 
&= \beta + p_i^{-} \xib/x_i - \gamma \xib 
\end{align*}


Similarly to \eqref{eq:x_bot},
\eqref{eq:bot-beta} is equivalent to writing $\xib' = \beta + p_i^{-} \xib^2/\xit - \gamma \xib$.
We now justify why we add the term $\beta x_i$ into the production rate of \eqref{eq:new-top}:
Since we are adding a term to \eqref{eq:x_bot} we must introduce an additive term, call it $\alpha$, 
to \eqref{eq:x_top} to cancel out the effect and make sure  $\frac{d(\xit/\xib)}{dt} = p_i^+ - p_i^-$. Below we show $\alpha = \beta x_i$.
\begin{eqnarray}
\frac{d \xit}{dt} &=& 
\alpha + p_i^{+} \xib - \gamma \xit
\label{eq:top-alpha}
\\
\frac{d \xib}{dt} &=& 
\beta + p_i^{-} \xib/x_i - \gamma \xib 
\label{eq:bot-beta}
\end{eqnarray}

\begin{align*}
    \frac{d(\xit/\xib)}{dt} &= \frac{\xit'\xib - \xit\xib'}{\xib^2} && \text{quotient rule}
    \\
    &= \frac{\xit'}{\xib} - \frac{\xit \xib'}{\xib^2}
    \\
    &= \frac{\alpha + p_i^{+} \xib - \gamma \xit}{\xib} - \frac{\xit(\beta + p_i^{-} \xib/x_i - \gamma \xib)}{\xib^2}
    && \text{substituting \eqref{eq:top-alpha}, \eqref{eq:bot-beta}}\\
    &=\frac{\alpha}{\xib} + p_i^{+} - \frac{\gamma \xit}{\xib} - p_i^{-} + \frac{\gamma \xit}{\xib} - \frac{\xit \beta}{\xib^2}
    \\
    &=
    \frac{\alpha \xib - \xit \beta}{\xib^2} + p_i^{+} - p_i^{-}
\end{align*}
For the above expression to equal $p_i^{+} - p_i^{-}$ as required,
we must have $\alpha \xib - \beta \xit = 0$, i.e., $\alpha = \beta \frac{\xit}{\xib} = \beta x_i$,
completing the justification for adding the term $\beta x_i$ in \eqref{eq:new-top}.

Now we explain why the addition of the $\beta$ term to $\xib$ prevents the problem identified above.
First, we claim that each $\xib$ term is bounded-positive.
Looking at \eqref{eq:new-bot},
we have two cases.
In the first case, 
if $\xib$ starts above $\beta / \gamma$, then $\xib$ cannot go under $\beta / \gamma$, 
since its production rate is at least $\beta$ and its degradation rate is $\gamma \xib$.
In the second case assume $\xib < \beta / \gamma$;
then $\xib' > 0$, i.e., it will increase towards $\beta / \gamma$. 
Thus $\xib$ is unconditionally bounded-positive.

We now want to establish that $\xib$ is also bounded above,
i.e., $\xib$ is bookended, provided that $\gamma$ is sufficiently large.
Write equation \eqref{eq:new-bot} as 
\[
\frac{d \xib}{dt} = \beta + \xib (p_i^{-} /x_i - \gamma) = \beta - b(t) \xib 
\]
where we define $b(t) = \gamma - p_i^{-} /x_i$.
If we ensure that $b(t) > b_\mathrm{min}$ for some $b_\mathrm{min} > 0$, 
then if ever 
$\xib > \beta / b_\mathrm{min} > \beta / b(t)$, then $\xib' < 0$, so it will decrease.
This
establishes that $\xib$ is bounded above.

It remains to establish that $b(t) > b_\mathrm{min}$ for some $b_\mathrm{min} > 0.$
We now do a case analysis.
For the first case, assume that all the variables in $\rho$ are bounded-above.
As long as $p_i^- / x_i$ is also bounded above,
i.e., for some constant $c$,
$p_i^- / x_i < c$,
then we can set $\gamma > c$ to ensure $b(t) = \gamma - p_i^-/x_i > \gamma - c = b_\mathrm{min} > 0.$

How can we ensure that $p_i^- / x_i$ is bounded above?
First, let us consider the case that the original system of polynomial ODEs ($P$) is in Hungarian form for $x_i$.
Then $p_i^{-} /x_i = q_i^-$ for some nonnegative polynomial $q_i^-$.
Since we have assumed that all the variables in $\rho$ are bounded-above,
$q_i^-$ is bounded-above as well, and we can set $\gamma$ larger than the 
maximum value ever taken on by any $q_i^-$ 
for $1 \leq i \leq n$.
This choice of $\gamma$ is sufficient to ensure that $b(t) > b_\mathrm{min}$.

On the other hand, consider the case that $P$ is not Hungarian form for $x_i$. 
The problem is that if $x_i$ approaches zero, $p_i^{-} /x_i$ may approach infinity since,
unlike the Hungarian case,
$x_i$ is not a factor in $p_i^-$ canceling $1/x_i$. 
Thus we explicitly disallow non-Hungarian variables to approach 0, 
as stated in the condition of the theorem.
Similarly to the above, all variables being bounded above implies $p_i^-$ is also bounded above.
Since $x_i$ cannot approach 0,
this means that $p_i^- / x_i$ is bounded above,
and we similarly choose $\gamma$ to be greater than its maximum value.

Note that because each $\xib$ is bookended, 
by our assumption that all $x_i$ are bounded-above,
\Cref{obs:bdd-x-xt} establishes that each $\xit$ is also bounded above,
obeying the constraints to make this a stable ratio-implementation by \Cref{def:bounded-implementation}.
\opt{sub,final}{\qed}
\end{proof}

We now argue that by slightly strengthening the conditions on the variables (requiring all variables, even Hungarian, not to start at 0 or approach 0)
we can assume the original ODE system has Laurent polynomials.

\begin{thmrep}
\label{thm:ODEtoTN_laurent}
Let $P$ be a system of Laurent polynomial ODEs over the set of variables $\calX$.
Then there is a transcriptional network $\calF(\gamma)$, constructed by Algorithm~\ref{alg:construction} with $p_i^+$ and $p_i^-$ being nonnegative Laurent polynomials,
such that for every upper-bound $b_\mathrm{max} > 0$
and every lower-bound $b_\mathrm{min} > 0$,
there is $\gamma > 0$ such that
for every trajectory $\rho$ of $P$ that is bookended by $b_\mathrm{max}$ and $b_\mathrm{min}$ on $\calX$,
$\calF(\gamma)$ stably ratio-implements $\rho$.
\end{thmrep}

\begin{proof}
    The proof is very similar to \Cref{thm:ODEtoTN},
    and the construction, \Cref{alg:construction}, is identical,
    since \Cref{alg:construction} expresses the rates as ratios of the polynomials $p^+$, $p^-$, and the transcription factor concentrations, which remain Laurent polynomials even if $p^+$, $p^-$ are Laurent polynomials to begin with.
    The extra hypothesis on variables not starting at or approaching 0 handles the fact that, for a Laurent monomial with a nontrivial denominator $\frac{1}{d}$, if $d$ is 0, then the monomial would be undefined.
\opt{sub,final}{\qed}
\end{proof}

\begin{remark}
Our main construction in \Cref{alg:construction} replaces each variable $x$ in the polynomial ODEs with a pair of transcription factors (variables in a new set of Laurent polynomial ODEs) $\tp{X},\bt{X}$, 
such that for all $t$, $x(t) = \frac{\tp{x}(t)}{\bt{x}(t)}$,
the so-called ``ratio representation'' of $x$.
It is also possible to inter-convert between these formats \emph{without} replacing variables.

We have $\tp{x}$ and $\bt{x}$ in a transcriptional network,
but there is no transcription factor that directly represents the value $x$.
Suppose we want a third transcription factor to \emph{directly} represent the original value $x$, in addition to representing $x$ indirectly via the ratio $\frac{\tp{x}}{\bt{x}}$.
Then we can introduce a new transcription factor $\hat X$ with ODE $\hat x' = \gamma(\frac{\tp{x}}{\bt{x}} - \hat x)$.
The larger is $\gamma$, the more closely $\hat x$ tracks $x = \frac{\tp{x}}{\bt{x}}.$
(Recall from the proof of \Cref{thm:ODEtoTN} that $\gamma$ must be bounded away from 0, but it does not need to be bounded above,
so we are free to set $\gamma$ arbitrarily large to make $\hat x$ track $x$ arbitrarily closely.)
Note that this remains a transcriptional network, since the positive term is a Laurent monomial and the negative term is $-\gamma \hat x$ as required.
We employ a similar trick in \Cref{sec:extremum}, introducing a transcription factor $x$ to track the value $\frac{\tp{z}}{\bt{z}} + a \frac{\tp{p}}{\bt{p}}.$

Conversely, if we have a single variable $x$ in a transcriptional network, we can let 
$\tp{x}(0)=x(0)$,
$\bt{x}(0)=1$,
with
$\tp{x}' = x'$
and
$\bt{x}' = \gamma-\gamma \bt{x}$,
which maintains the stronger condition that $x(t) = \frac{\tp{x}(t)}{\bt{x}(t)}$ for all $t$
(i.e., $\frac{\tp{x}}{\bt{x}}$ ``tracks'' $x$ with zero lag).
\end{remark}

\begin{remark}
\label{rem:eqs-need-simplification}
When applying the construction in \cref{alg:construction}, algebraic simplification should be performed on the right-hand side of \eqref{eq:new-top} and \eqref{eq:new-bot}, reducing all rational expressions to lowest terms. This prevents having redundant fractions such as $\frac{\xib}{\xib}$, which would imply that $\xib$ acts simultaneously as both an activator and a repressor. 
Moreover, when a variable $x_i$ is permitted to approach 0 (i.e., when $x_i$ is Hungarian; see \Cref{thm:ODEtoTN}), this simplification is essential as shown in the next remark.
\end{remark}

\begin{remark}
\label{rem:repressors-do-not-approach-zero}
Note that a Hungarian form polynomial ODE can be written 
$x_i' = p_i^+ - x_i q_i^-$,
where $p_i^+$ and $q_i^-$ are arbitrary nonnegative polynomials in $x_1,\dots,x_n$
($q_i^-$ being $p_i^-$ with the common $x_i$ factored out.)
In this case equation \eqref{eq:new-bot} simplifies to
$\xib' = \beta + q_i^- \xib - \gamma \xib$.
Note that $\xit$ no longer appears in the denominator; consequently, even when $x_i$ (and hence $\xit$) approaches 0, 
$\xit$ is not a vanishing repressor in the actual transcriptional network, preserving the Hill-function regime of \Cref{sec:justification-laurent} (see \Cref{rem:proprtional}).
\end{remark}

\section{Simplification of transcriptional network limiting activators and repressors per promoter}
\label{sec:simplification}

\begin{toappendix}
\label{apx:simplification}
\end{toappendix}







\todo{DD: implement this in ode2tn}

The model of Laurent monomials to describe the effect of these on production rate may break down with a large number of activators and repressors on a single promoter site (see \Cref{sec:justification-laurent}).
Thus a disadvantage of the construction of \Cref{alg:construction} is that each transcription factor promoter can end up regulated by many activators and many repressors, corresponding to having many terms in the numerator and/or denominator of a single Laurent monomial,
possibly pushing the model too far.
To mitigate this concern,
in this section we describe a way to simplify the construction such that each promoter is controlled by at most one repressor and at most two activators, i.e., each Laurent monomial will be of the form $\alpha \frac{a_1 a_2}{r}$ for two activators $a_1,a_2$ and a repressor $r$
(possibly some of $a_1, a_2, r$ may be missing).

We use a technique motivated by the methods of polynomial 
quadratization~\cite{bychkov2021optimal}---converting polynomial ODEs to have degree at most two.
Importantly, however, we cannot simply compose the existing quadratization constructions with ours.
Indeed,
simply running the quadratized polynomial ODEs through our construction,
would result in Laurent monomials of numerator degree $4$ and denominator degree $3$ (e.g., if the $p_i^-$ term is the degree $2$ monomial $x_j x_k$, then the middle monomial of \eqref{eq:new-bot} would be 
$
\frac{\tp{x_j}}{\bt{x_j}}
\frac{\tp{x_k}}{\bt{x_k}}
\frac{\bt{x_i}^2}{\tp{x_i}}
$).
Instead, we must carefully set up the initial polynomial ODE transformation to ensure cancellation of terms within the construction of \Cref{alg:construction} will limit the degree of the numerator to at most $2$ and denominator to at most $1$.

We note that the conditions required on the ODEs of \Cref{thm:simplification} are more stringent than those of \Cref{thm:ODEtoTN}, requiring all variables to start positive and stay bounded away from 0 
(bounded-positive, \Cref{def:bounded-positive}).
\opt{full}{This avoids some divide-by-zero problems in the construction.}
\opt{sub,final}{
This avoids some divide-by-zero problems in the construction, as shown in the full proof given in \Cref{apx:simplification}.
}



\begin{thmrep}
\label{thm:simplification}
Let $P$ be a system of polynomial ODEs over the set of variables $\calX$. 
Then there is a transcriptional network $\calF(\gamma)$, 
whose every positive term is a Laurent monomial of the form
    $\alpha \frac{a_1 a_2}{r}$,
    where each $a_1,a_2,r$ is a transcription factor or the constant $1$.
such that for every upper-bound $b_\mathrm{max} > 0$
and every lower-bound $b_\mathrm{min} > 0$,
there is $\gamma > 0$, such that  
    for every trajectory $\rho$ of $P$ that is bookended by $b_\mathrm{max}$ and $b_\mathrm{min}$ on $\calX$,
    $\calF(\gamma)$ stably ratio-implements $\rho$.
\end{thmrep}

\begin{proof}
Here we take advantage of the fact that \Cref{alg:construction} actually can process arbitrary Laurent polynomials (not merely polynomials) (see \Cref{thm:ODEtoTN_laurent}), 
since we first modify $P$ to a system of ODEs with Laurent polynomial right-hand sides.

Consider a polynomial ODE system 
\begin{equation}
\label{eq:poly-ode-xi-to-simplify}
\frac{d x_i}{dt}=p_i(x_1,\dots,x_n)
= \sum_{j=1}^{k_i} \alpha_{i,j} M_{i,j}
\end{equation}
where each $M_{i,j}$ is a monomial in the variables $x_1,\dots,x_n$.
For all monomials $M_{i,j}$,
we introduce new variables $z_{i,j}$ and set their ODE to ensure that at all times $t \geq 0$,
\begin{align}
\label{eq:define-z}
z_{i,j}(t)=\frac{x_i(t)}{M_{i,j}(t)}.
\end{align}
Since each $x_i(0) > 0$, we have $M_{i,j}(0)>0$ and $z_{i,j}(0)>0$.
Then 
\eqref{eq:poly-ode-xi-to-simplify} can be rewritten as
\begin{align}
\frac{d x_i}{dt}
&=
\sum_{j=1}^{k_i}\alpha_{i,j}\frac{x_i}{z_{i,j}}
=
x_i \sum_{j=1}^{k_i} \frac{\alpha_{i,j}}{z_{i,j}}.
\label{eq:x-system}
\end{align}

We now derive the ODE equations for $z_{i,j}$ via the logarithmic derivative:
for any positive differentiable function $f$, the quantity
$\dot f/f = \frac{d}{dt} \ln f$ converts products and quotients into sums and
differences.
Write $M_{i,j} = x_{1}^{e_1}\cdots x_{n}^{e_n}$.
Then applying the identity $\dot f/f = \frac{d}{dt} \ln f$ a few times, 
\begin{align}
\frac{\dot z_{i,j}}{z_{i,j}}
&
= \frac{d}{dt} \ln z_{i,j}
= \frac{d}{dt}\ln \frac{x_i}{M_{i,j}}
= \frac{d}{dt} \ln x_i - \frac{d}{dt} \ln M_{i,j}
= \frac{\dot x_i}{x_i} - \frac{d}{dt} \ln M_{i,j}
\nonumber
\\&
= \frac{\dot x_i}{x_i} - \frac{d}{dt} \ln (x_{1}^{e_1}\cdots x_{n}^{e_n})
= \frac{\dot x_i}{x_i} - \frac{d}{dt} \sum_{\ell=1}^n e_\ell \ln x_{\ell}
= \frac{\dot x_i}{x_i} - \sum_{\ell=1}^n e_\ell \frac{d}{dt} \ln x_{\ell}
\nonumber
\\&
= \frac{\dot x_i}{x_i} - \sum_{\ell=1}^n e_\ell \frac{\dot x_{\ell}}{x_{\ell}}.
\label{eq:zdot_over_z}
\end{align}

Dividing \eqref{eq:x-system} by $x_i$ 
(and changing $j$ to $m$) shows that
$\frac{\dot x_i}{x_i} 
= \sum_{m=1}^{k_i} \frac{\alpha_{i,m}}{z_{i,m}}$, and likewise for each
$x_{\ell}$.
Substituting this into the right side of \eqref{eq:zdot_over_z} and multiplying through by $z_{i,j}$ yields
\begin{align}
\frac{d z_{i,j}}{dt}
&=z_{i,j}\left[
 \sum_{m=1}^{k_i} \frac{\alpha_{i,m}}{z_{i,m}}
 - \sum_{\ell=1}^n e_\ell
  \sum_{m=1}^{k_{\ell}} 
  \frac{\alpha_{\ell,m}}{z_{\ell,m}}
\right],
\label{eq:z-system}
\end{align}
Since all variables are bookended in the trajectory by assumption, equation \eqref{eq:define-z} implies that $z_{i,j}$ is bookended as well.
Importantly, each term of $\dot z_{i,j}$ in \eqref{eq:z-system} is simply $\lambda \frac{z_{i,j}}{z_{\ell,m}}$ for some constant coefficient $\lambda$;
this will be key in transforming this ODE system to a transcriptional network with a limited number of activators and repressors.

If the initial conditions satisfy $z_{i,j}(0)=x_i(0)/M_{i,j}(0)$ for all $i,j$,
then these equalities are preserved for all time:
the derivation above shows that $x_i/M_{i,j}$ satisfies the same ODE as
$z_{i,j}$, so uniqueness of solutions implies
$z_{i,j}(t)=x_i(t)/M_{i,j}(t)$ for all $t$.
Hence along every such trajectory $M_{i,j}=x_i/z_{i,j}$,
and \eqref{eq:x-system} reproduces the original polynomial system exactly.


Now collect all variables $x_i$ and $z_{i,j}$ into a single list
\[
y_1,\dots,y_N.
\]

Then the system described by \eqref{eq:x-system} and \eqref{eq:z-system} can be written in the form
\begin{align}
\frac{d y_a}{dt}
&=
\sum_{b=1}^N \lambda_{a,b}\frac{y_a}{y_b},
\qquad a=1,\dots,N,
\label{eq:y-system}
\end{align}
for suitable real constants $\lambda_{a,b}$.

For each $a,b$, define
\[
\lambda_{a,b}^+ = \max\{\lambda_{a,b},0\},
\qquad
\lambda_{a,b}^- = \max\{-\lambda_{a,b},0\}.
\]
Then $\lambda_{a,b} = \lambda_{a,b}^+ - \lambda_{a,b}^-$, where
$\lambda_{a,b}^+,\lambda_{a,b}^- \ge 0$, and we have

\begin{align}
\frac{d y_a}{dt}
&=
\sum_{b=1}^N \lambda_{a,b}^+\frac{y_a}{y_b}
-
\sum_{b=1}^N \lambda_{a,b}^-\frac{y_a}{y_b}.
\label{eq:y-system-split}
\end{align}

For each variable $y_a$, include two transcription factors $Y_a^\top$ and $Y_a^\bot$
with positive initial concentration such that
$
y_a=\frac{Y_a^\top}{Y_a^\bot}.
$
Construct the transcriptional network from \eqref{eq:y-system-split} by applying \Cref{alg:construction}:
\begin{align}
\frac{d Y_a^\top}{dt}
&=
\beta \frac{Y_a^\top}{Y_a^\bot}
+
\sum_{b=1}^N \lambda_{a,b}^+\frac{Y_a^\top Y_b^\bot}{Y_b^\top}
-
\gamma Y_a^\top
\label{eq:Y-top}
\\[1ex]
\frac{d Y_a^\bot}{dt}
&=
\beta
+
\sum_{b=1}^N \lambda_{a,b}^-\frac{Y_a^\bot Y_b^\bot}{Y_b^\top}
-
\gamma Y_a^\bot.
\label{eq:Y-bot}
\end{align}


Here $\beta$ is any positive constant, and $\gamma$ is a positive constant chosen once for the whole network and trajectory as in \Cref{alg:construction}.
Every positive term appearing in \eqref{eq:Y-top}--\eqref{eq:Y-bot} is of one of the forms
\[
\sigma,\qquad \sigma u,\qquad \frac{\sigma u}{v},\qquad \frac{\sigma uv}{w},
\]
for constant $\sigma$ and transcription factor concentrations $u,v,w$,
completing the construction.
\opt{sub,final}{\qed}
\end{proof} 

\section{Examples}
\label{sec:examples}

We illustrate our main construction 
(\Cref{eq:new-top,eq:new-bot}) from \Cref{alg:construction} on several example systems of polynomial ODEs. Code producing the plots is available at \url{https://github.com/UC-Davis-molecular-computing/ode2tn/blob/main/notebook.ipynb}.
\opt{sub,final}{
    Additional examples are shown in \Cref{sec:app-examples}.
}

Unless explained otherwise, the plots of variables such as $x$ are showing the value $\frac{\tp{x}}{\bt{x}}$ over time for the underlying transcription factors $\tp{x}$ and $\bt{x}$.

\subsection{Sine-cosine oscillator}
\label{sec:sin-cosine-example}

\begin{figure}
\centering
\includegraphics[width=1\linewidth]{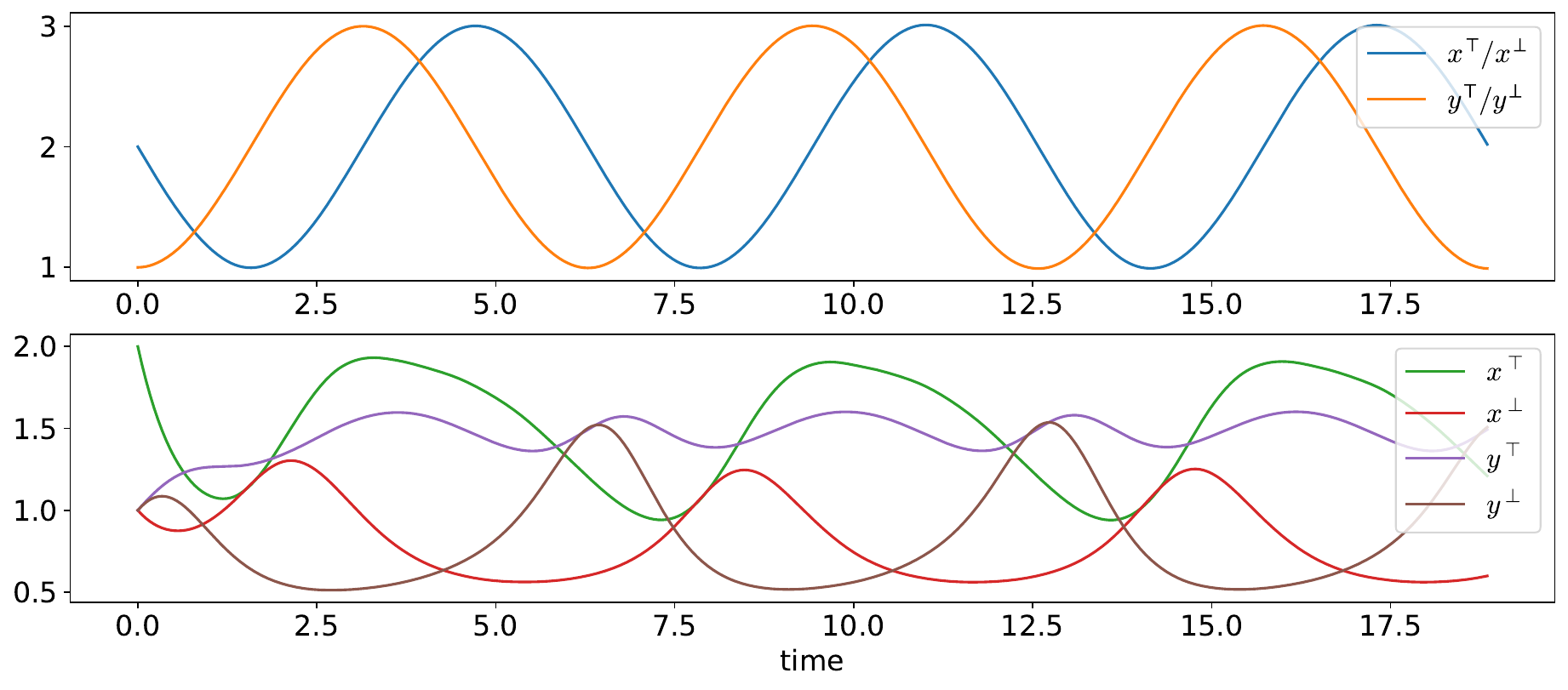}
\caption{
    Plot of ``shifted positive'' sine-cosine oscillator implemented with a transcriptional network.
    There are two subplots: the first shows the ``original'' variables $x,y$ (i.e., the computed ratios $\tp{x}/\bt{x}$ and $\tp{y}/\bt{y}$).
    The second subplot shows underlying transcription factor variables 
    $\tp{x},\bt{x},\tp{y},\bt{y}$.
}
\label{fig:sine-cosine-plot}
\end{figure}

The sine-cosine oscillator is the system
\begin{align*}
x' &= y
\\
y' &= -x
\end{align*}
since we can write $x' = \sin(t)' = \cos(t) = y$ and $y' = \cos(t)' = - \sin(t) = -x.$ 
However, this takes negative values, whereas transcription factor concentrations must be nonnegative.
This could be solved with the so-called \emph{dual-rail} representation~\cite{chen2023rate} where negative values $x$ can be represented as the difference of two positive values $x = x^+ - x^-$.
However, it is simpler in this case simply to ``shift'' the oscillator upwards to stay positive.
The following system, starting with $x=2,y=1$ (i.e., initial conditions where each variable starts 2 higher than originally),
will oscillate between a maximum amplitude of 3 and a minimum amplitude of 1;
see \Cref{fig:sine-cosine-plot}.
It works by replacing variable $v$ by $v-2$ wherever it appears in the ODEs.
\begin{align*}
x' &= y-2
\\
y' &= -x+2
\end{align*}
Recall we must choose a constant $\gamma$ to apply the construction of \Cref{alg:construction}, which must be greater than the maximum value taken by any $p_i^- / x_i$.
Here we think of $x_1 = x$ and $x_2 = y$.
In this case, $p_1^- = 2$ (negative term for $x'$) and $p_2^- = x$ (negative term for $y'$).
We have $2/x \le 2$ 
(the case $i=1$) and 
it turns out that 
$x/y < 2.5$
for all values in the trajectory of this oscillator.
Thus choosing $\gamma = 2.5$ suffices to ensure all transcription factors stay bounded.
In subsequent examples we omit a detailed analysis and simply pick a $\gamma$ that appears empirically to keep the transcription factors bounded;
note that the argument of the proof of \Cref{thm:ODEtoTN} gives a sufficient but not necessary condition for choosing $\gamma$;
for example,
although $x/y$ does sometimes exceed 2,
it appears in practice that setting $\gamma=2$ also keeps the transcription factors bounded.

After applying the construction of \Cref{alg:construction},
i.e., using \Cref{eq:new-top,eq:new-bot} to generate the ODEs for $\tp{x}', \bt{x}', \tp{y}',$ and $\bt{y}'$ based on the ODEs for $x'$ and $y'$ above (where we can take $x$ above to be $x_1$, and take $y$ above to be $x_2$, in the proof of \Cref{thm:ODEtoTN}),
this generates the transcriptional network:
\begin{align*}
\tp{x}' &= \frac{\tp{x}}{\bt{x}} + \frac{\bt{x}\tp{y}}{\bt{y}} - 2.5\tp{x}
&
\bt{x}' &= 1 + 2\frac{\bt{x}^2}{\tp{x}} - 2.5\bt{x} 
\\
\tp{y}' &= \frac{\tp{y}}{\bt{y}} + 2\bt{y} - 2.5\tp{y} 
&
\bt{y}' &= 1 + \frac{\tp{x}\bt{y}^2}{\bt{x}\tp{y}} - 2.5\bt{y} 
\end{align*}
and initial values
$
\tp{x}(0) = 2,
\quad
\bt{x}(0) = 1,
\quad
\tp{y}(0) = 1,
\quad
\bt{y}(0) = 1.
$

\begin{toappendix}
\label{sec:app-examples}
For subsequent examples, we will not show the values of the underlying $\tp{v},\bt{v}$ variables,
but for the sine-cosine system, they can also be seen in \Cref{fig:sine-cosine-plot}.

\subsection{Bubble sort}

\begin{figure}[t]
\centering
\includegraphics[width=\linewidth]{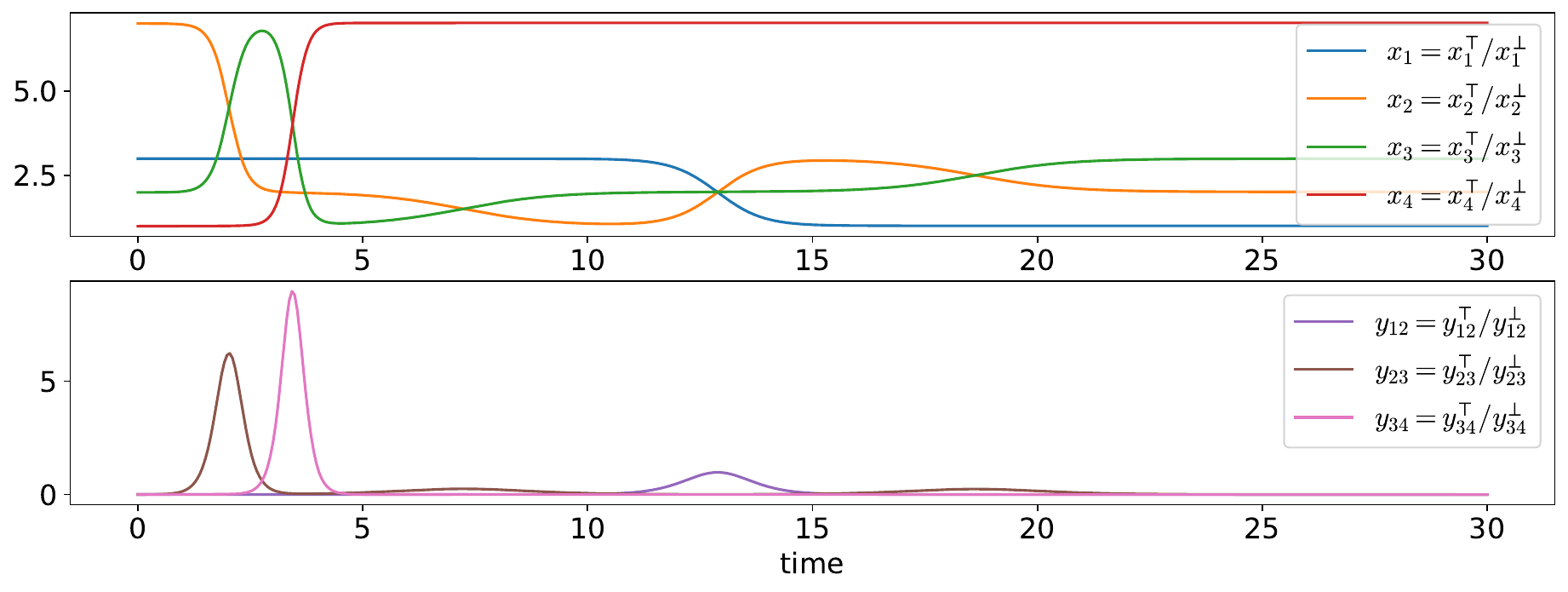}
\caption{Bubble sort system with four variables $x_1 = 3$, $x_2 = 7$, $x_3 = 2$, $x_4 = 1$ and we start $y_{1,2} = y_{2,3} = y_{3,4} = 0.01$.
The $y_{i,i+1}$ value spikes when it is swapping $x_i$ and $x_{i+1}$.
(Though there are conditions under which the plot does not look like discrete swaps separated in time.)
As in \Cref{fig:sine-cosine-plot}, each variable is the computed ratio of two transcription factors.
}
\label{fig:bubble_sort}
\end{figure}

This is a system introduced by Paul and H\"{u}per~\cite{paul1993analog}, which implements a sorting algorithm reminiscent of the discrete bubble sort.
The goal is to sort the values $x_1,\dots,x_n$.
Variables $y_{i,i+1}$ are intended to swap the values of $x_i$ and $x_{i+1}$ if they are out of order.\footnote{
    That said, this is not exactly an implementation of bubble sort, since swaps do not occur in discrete sequential steps.
    For instance, if started with values in reverse order such as $x_1 = 10, x_2 = 9, \dots, x_{10} = 1$,
    the ODEs will change all values in what appears to be one ``step''.
}


The ODEs for this system are
\begin{eqnarray*}
x_1' &=& - y_{1,2} \\
x_2' &=& y_{1,2} - y_{2,3} \\
x_3' &=& y_{2,3} - y_{3,4} \\
&\vdots&\\
x_{n-1}' &=& y_{n-2,n-1} - y_{n-1,n} \\
x_n' &=& y_{n-1,n}\\
y_{1,2}' &=& y_{1,2}(x_1 - x_2)\\
y_{2,3}' &=& y_{2,3}(x_2 - x_3)\\
&\vdots&\\
y_{n-1,n}' &=& y_{n-1,n}(x_{n-1} - x_n)
\end{eqnarray*}

The initial values of the $x_i$'s represent the values that need to be sorted, and the initial values of all $y_{i,i+1}$'s are $\varepsilon > 0$. Since these ODE's are not in Hungarian form, we assume that $x_i$'s have  positive initial values.  
\cref{fig:bubble_sort} shows an example after we apply the construction given in  \Cref{alg:construction}.





\subsection{Schl\"{o}gl system}
\label{sec:schlogl}

\begin{figure}
\centering
\includegraphics[width=0.65\linewidth]{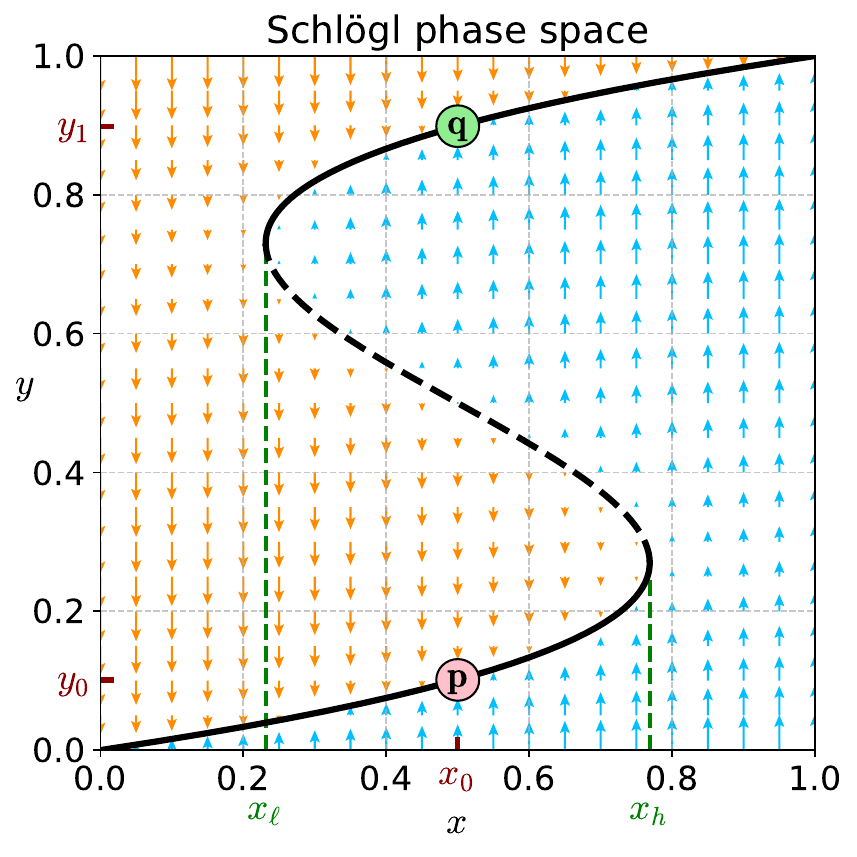}
\vspace{-0.4cm}
\caption{\footnotesize
    Phase-space diagram of the Schl\"{o}gl CRN of \Cref{sec:schlogl}.
    The black curve is the plot of the function 
    $x = f(y) = 11 y^3 - 16.5 y^2 + 6.5 y,$
    mirrored around the line $y=x.$
    No reaction changes $X$; we think of $X$ as ``externally controlled'' as a way to influence $Y$ through the reactions. 
    For any point $(x,y)$ in the plane, the vector arrows show the direction of $y'$, pushing $Y$ up if to the right of the curve (blue arrows) and down otherwise (orange arrows).
    Imagine for example we start at point $\vp = (x_0,y_0)$.
    If we then move $X$ above the value $x_h$, $Y$ converges to the upper solid part of the black curve.
    If we subsequently lower the value of $X$ to its original coordinate $x_0$, $Y$ will converge at a different point $\vq = (x_0,y_1)$ above $\vp$, remembering that we changed $X$,
    despite returning $X$ to its original value.
    Symmetrically, if we lower the value of $X$ below the threshold $x_\ell$, 
    and raise it back again to its original value, we restore the system back to point $\vp$.
    This dynamic process is shown in \Cref{fig:schlogl-plot}.
    The dashed portion of the curve shows unstable fixed points of the CRN; the solid portions are stable fixed points.
}
\label{fig:schlogl-phase-space}
\end{figure}

\begin{figure}
\centering
\includegraphics[width=1\linewidth]{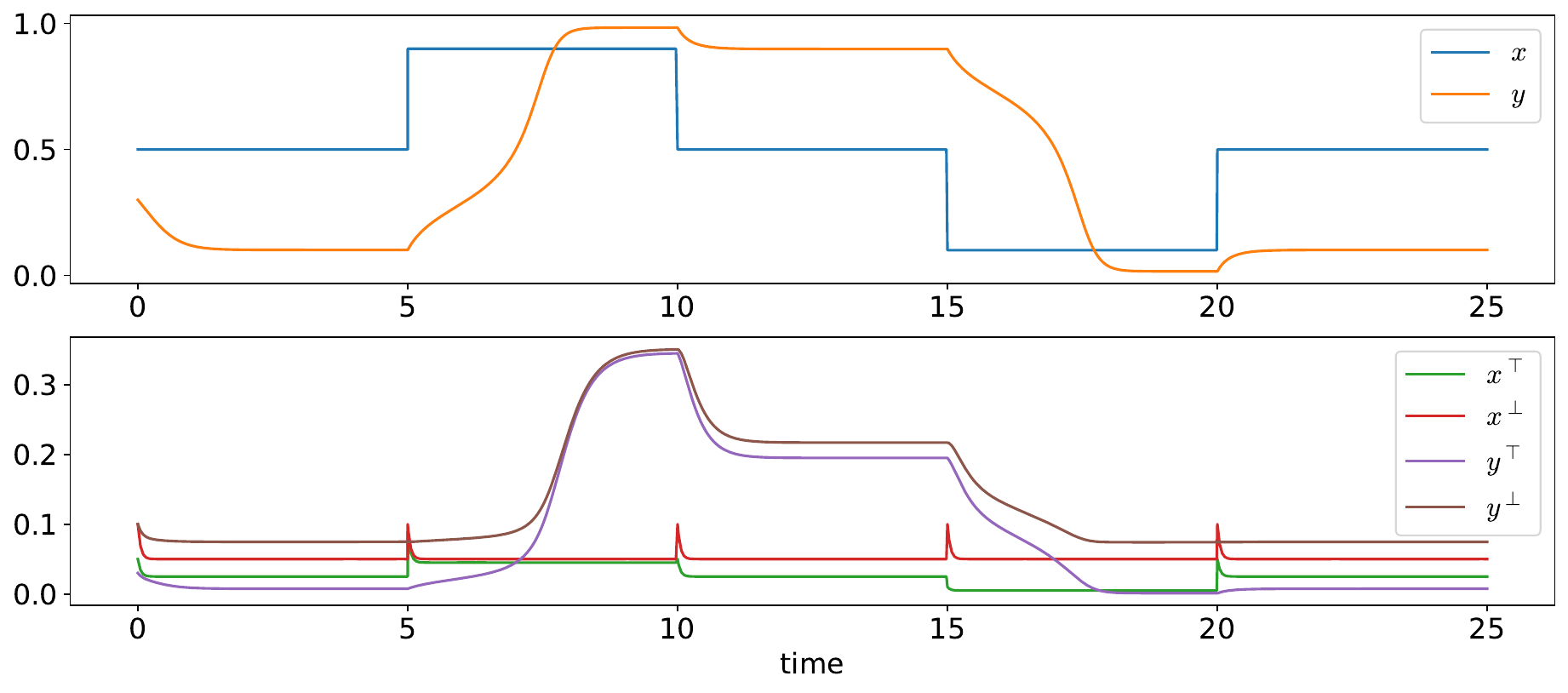}
\vspace{-0.5cm}
\caption{
\footnotesize
    Plot of Schl\"{o}gl system with parameters $k_1=11, k_2=16.5, k_3=6.5$, with ``resets'' of the $x$ variable.
    $x$ starts at 0.5 ($x_0$ in \Cref{fig:schlogl-phase-space}).
    $y$ starts ``small'', which makes $y$ converge to $\approx 0.1$ ($y_0$ in \Cref{fig:schlogl-phase-space}).
    At time 5 we set $x$ ``large'' ($=0.9$, by setting $\tp{x}=0.09$ and $\bt{x}=0.1$), and $y$ follows (reaching $y\approx 0.984$ such that $f(y)=0.9$).
    We then set $x$ back to $0.5$ at time 10, which causes $y$ to converge to $y_1 \approx 0.9$.
    We then set $x$ below $x_\ell$ at time 15, then raise it back to 0.5 for a second time at time 20, and the system remembers that we previously made $x$ low, 
    returning $y$ to low ($\approx 0.1$). In the top plot, $x$ and $y$ are the computed ratios of the transcription factors.
}
\label{fig:schlogl-plot}
\end{figure}

The Schl\"{o}gl chemical reaction network~\cite{schlogl1972chemical} is given by the reactions
\begin{align*}
X &\rxn^{1} X+Y
\\
3Y &\revrxn^{k_1}_{k_2} 2Y
\\
Y &\rxn^{k_3} \emptyset
\end{align*}
Note: this network is presented differently in~\cite{schlogl1972chemical}.
There, our species $Y$ is called $X$, and our species $X$ is called $C$.
Also, unit rate constants are assumed, so the rate constants $k_1,k_2,k_3$ above are imitated by extra reactants that are assumed to be held at a fixed concentration despite being consumed in each reaction.

This corresponds to the ODE
$
y' = x - (k_1 y^3 - k_2 y^2 + k_3 y).
$
In our plots, $k_1=11, k_2=16.5, k_3=6.5$.
Note that the system does not change $x$.
One thinks of this system as implementing a ``1-bit memory'' (a.k.a.\ hysteresis) in the sense that if we imagine an external controller altering $x$, say by raising it and then lowering it back to its original value, then the value of $y$ converges to a different value than it was at before altering $x$.
This can be seen by inspecting the phase-space diagram in \Cref{fig:schlogl-phase-space}.
\Cref{fig:schlogl-plot} shows an example where we change the value of $x$ in a sequence where the new value of $y$ depends on whether $y$ was previously large (close to 1) or previously small (close to 0).

\subsection{Willamowski-R\"{o}ssler  system (chaotic)}\label{sec:chaotic}

\begin{figure}
\centering
\includegraphics[width=1\linewidth]{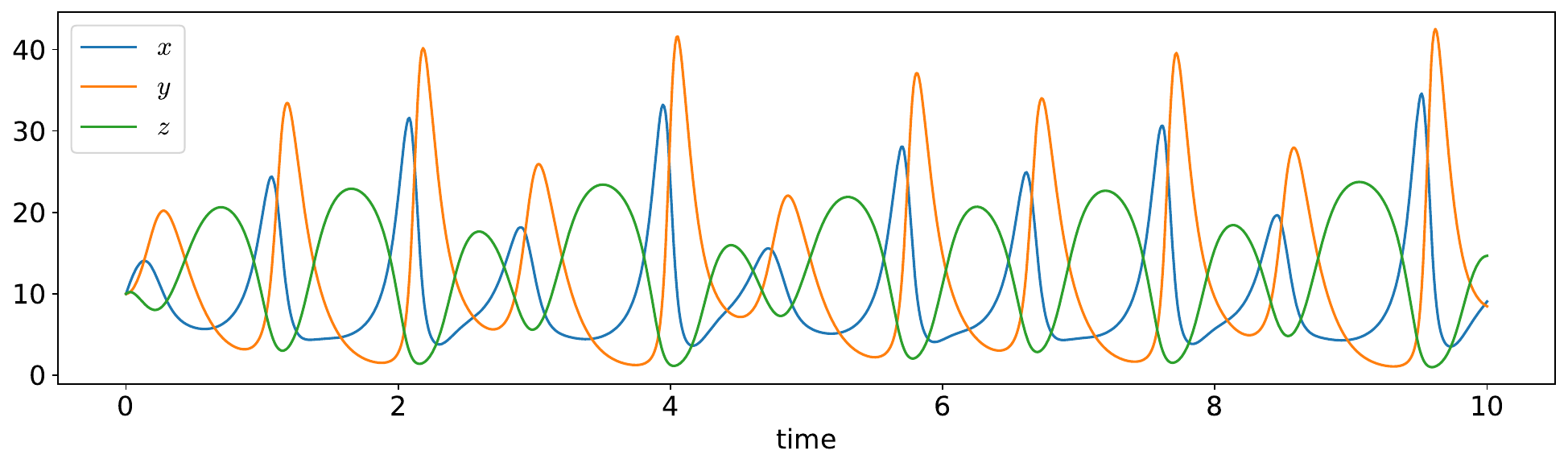}
\caption{
    Willamowski-R\"{o}ssler chemical reaction network with rate constants given in \cref{sec:chaotic} and initial concentrations $x(0)=10$, $y(0)=10$ and $z(0)=10$. $x$, $y$, and $z$ are the computed ratios of their respective transcription factors.
}
\label{fig:rossler-plot}
\end{figure}

The Willamowski-R\"{o}ssler system \cite{willamowski1980irregular} is a nonlinear dynamical system derived from the reactions
\begin{align*}
X   &\revrxn^{30}_{0.5} 2X
\\
X+Y &\rxn^1 2Y
\\
Y   &\rxn^{10} \emptyset
\\
X+Z &\rxn^1 \emptyset
\\
Z   &\revrxn^{16.5}_{0.5} 2Z
\end{align*}
This system exhibits chaotic behavior, meaning its trajectories are highly sensitive to initial conditions and can show unpredictable yet structured patterns over time.
See \Cref{fig:rossler-plot}.

\subsection{PID controller}

A proportional–integral–derivative (PID) controller is a feedback-based mechanism widely used to maintain continuous control and automatic adjustments in machines and processes. 
The PID controller consists of three components: proportional (P), integral (I), and derivative (D) control. As a simple example, consider a thermostat, which demonstrates how the PID controller uses these three elements to maintain the temperature at the desired set point.
Biomolecular integral and PID controllers are an active area in synthetic biology and molecular control~\cite{briat2016antithetic,aoki2019universal,filo2022hierarchy}.

\begin{figure}
\centering
\includegraphics[width=1\linewidth]{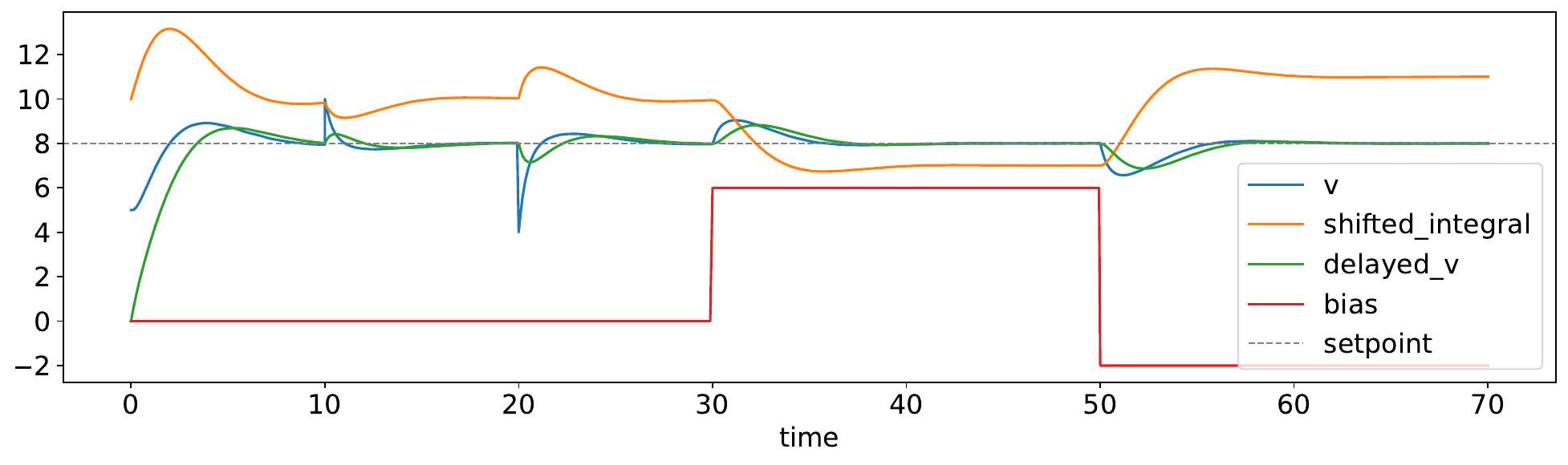}
\caption{
    PID controller, attempting to keep $v$ at a setpoint of 8 (called $\sigma$ in the main text), in response to external disturbances,
    with PID coefficients $P=1.5, I=1, D=1$.
    At time 10, we perturb $v$, increasing it to 10 (by setting $\tp{v}=10$ and $\bt{v}=1$), and at time 20, we perturb $v$ again, similarly decreasing it to 4,
    each representing an ephemeral disturbance to the system.
    At time 30, we perturb $v$ in a more permanent way, by adding a constant ``bias'' term of 6 to its ODE,
    which remains in $v$'s ODE for all times $30 \le t \le 50$,
    at which point we shift the bias to $-2$.
}
\label{fig:pid-plot}
\end{figure}

The controller first considers the difference (error) between the setpoint and the measured temperature. The output signal of this part is proportional to the size of the error.
The integral term sums the error over time, addressing steady-state errors that proportional control alone cannot resolve. For example, when the temperature consistently remains below the setpoint, the integral term increases the output to fix this problem.
The derivative term considers the rate of change of the error. This helps predict temperature trends and adjust the output signal to avoid overshoot. See \cref{fig:pid-plot}.

There are coefficients $P,I,D$ corresponding to the proportional, integral, and derivative terms.
Let the setpoint be $\sigma$ (a constant),
and let $v$ be the value we want to maintain at the setpoint.
The variable $i$ represents the integral of the error shifted up by the constant $\mu$ so that it stays non-negative; when we need the integral itself, we use the value $i-\mu$.
The derivative term is represented as follows.
First, $d$ represents a ``delayed $v$'', lagging behind $v$ by some amount controlled by the constant $\lambda$.
Therefore, the expression $v - d$ is an approximation of the derivative of $v$,
and $d-v$ represents the derivative of the error $\frac{d}{dt} (\sigma - v) = -\frac{dv}{dt}$.
The ODEs to achieve these are
\begin{align*}
v' &= P(\sigma-v) + I(i-\mu) + D(d - v)
\\
i' &= \sigma - v
\\
d' &= \lambda (v - d)
\end{align*}
We can start with $d(0)=0$, $v(0)$ being any positive value, and $i(0)=\mu$.

\subsection{Extremum seeking feedback scheme}
\label{sec:extremum}

This system originates with~\cite{krstic2000stability}.
This is an ``extremum-seeking'' system.
There is some ``objective function'' $f$ assumed to be unknown to us.
The goal is to adjust the value of a variable in the ODEs until it finds a local maximum of $f$.
\Cref{fig:extremum-seek-local-max-plot}
shows a function $f(x)$ with two local maxima.
The goal of the system is, starting with a variable $x$,
for $x$ to hill-climb to the nearest local maximum of $f$.

In our setting,
we think of $f$ as an unknown process that, based on the current concentration of some transcription factor $X$, produces a transcription factor $F$ characterized by some unknown function $f$.
In other words, if $X$ has concentration $x$, $F$ has concentration $f(x)$.
We imagine $F$ is some natural reagent
whose production we wish to maximize.
Our goal is to maximize the concentration of $F$ by adjusting $X$, without knowing how $F$ depends on $X$.
(i.e., the standard black-box optimization problem, but the computation of the optimization must be done by ODEs rather than by a standard algorithm).

\begin{figure}
\centering
\includegraphics[width=1\linewidth]{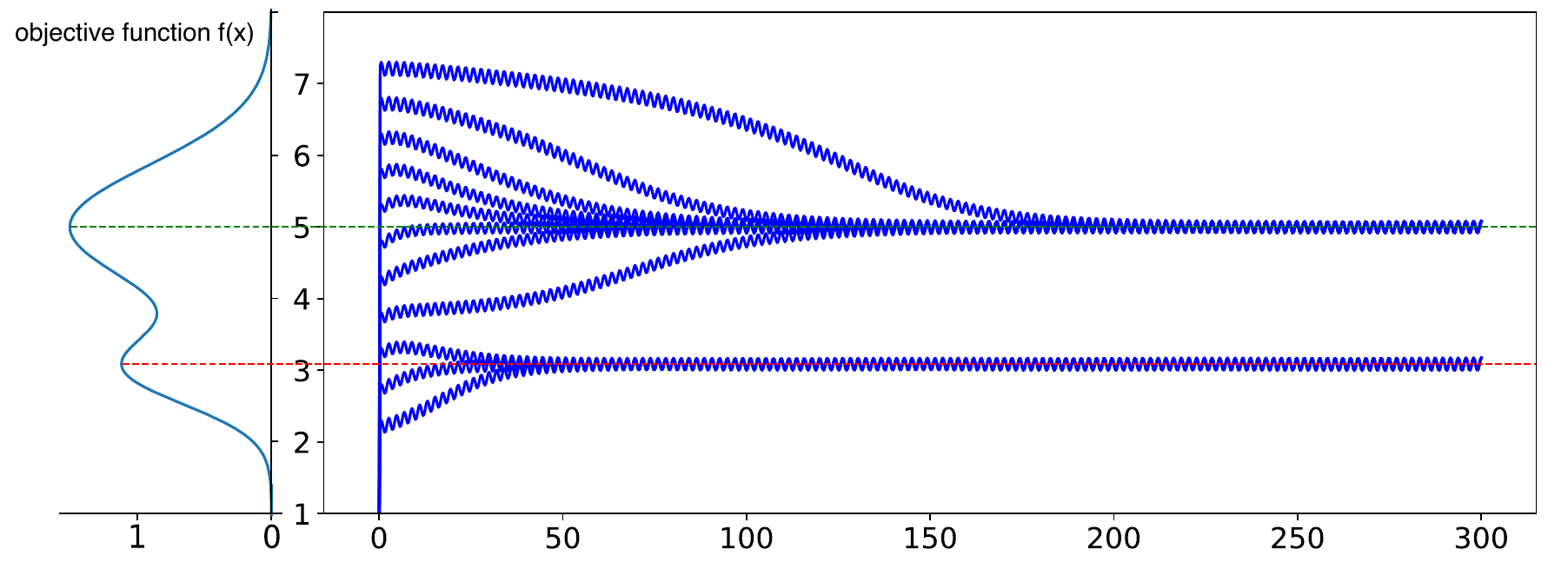}
\caption{
    Extremum seeking feedback scheme.
    The objective function is $f(x) = e^{-2(x-3)^2} + e^{-2(x-5)^2/3},$
    with local maxima at $x \approx 3.08$  and $x \approx 5$.
    ($f$'s graph shown rotated 90 degrees counterclockwise)
    This plot shows transcription factor $x$'s trajectory with various initial values between 2.5 and 7.5.
    (Outside of that interval, the slope of $f$ is sufficiently close to 0 to severely delay the time required for $x$ to find the maxima.)
}
\label{fig:extremum-seek-local-max-plot}
\end{figure}

The following ODEs achieve this by using a sine-cosine oscillator (on variables $p$ and $q$) to move $z+ap$ up or down from the current value of $z$,
with constants 
$\omega=3$ (period of oscillation), $\lambda=0.3$, 
$a=0.1$ (amplitude of oscillation),
$k=0.15$ (rate of convergence).
$x'$ is set so that $x$ tracks 
$z+ap$.
Intuitively, this works in the following way:
$p$ goes up and down from $+1$ to $-1$. 
When $p$ is positive, the $ap$ term in $x'$ drags $x$ up slightly 
(say to $x_\mathrm{h}$), and when $p$ is negative, the $ap$ term drags $x$ down slightly
(say to $x_\mathrm{l} < x_\mathrm{h}$).
If $f'(x) > 0$ (i.e., $f$ has positive slope, so we should increase $x$ to find the local maximum to the right),
then this means that $f(x_\mathrm{h})p$ when $p=1$ (the positive term in $w'$) will have larger magnitude than $|f(x_\mathrm{l})p|$ when $p$ is $-1$.
Thus, since the push to increase $w$ when $p$ is positive is stronger than the push to decrease $w$ when $p$ is negative, the net effect pushes $w$ higher.
Symmetry logic shows that when $f'(x) < 0$ (the slope is negative, so we should \emph{decrease} $x$ to reach the local maximum to the left),
the net effect will be to push $w$ lower.
\begin{align*}
p' &= \omega q
\\
q' &= -\omega p
\\
w' &= -\lambda w + f(x)p
\\
z' &= k w
\\
x' &= \gamma (z+ap - x)
\end{align*}
with initial values $p=0, q=1, w=0, x=0$,
and the initial value of $z$ represents where to start the search for a local maximum.

However, as with our simpler example with just the sine-cosine oscillator in \Cref{sec:sin-cosine-example},
\Cref{thm:ODEtoTN} only allows us to implement nonnegative trajectories.
We follow the same approach as \Cref{sec:sin-cosine-example},
shifting the oscillator implemented by $p$ and $q$ (as well as $w$, which also goes negative in the above system) up by 2 to keep it nonnegative, oscillating between a peak of 3 and a trough of 1, by shifting the initial values of $p,q,w$ up by 2.
Then, when we want to reference the original (possibly negative) value of $p$, $q$, or $w$, we reference the value minus 2.
We also similarly shift $z$ down by $2a$ (explained below):
\begin{align*}
p' &= \omega (q-2)
\\
q' &= -\omega (p-2)
\\
w' &= -\lambda (w-2) + f(x)(p-2)
\\
z' &= k (w-2)
\\
x' &= \gamma (z + 2a + a(p-2) - x) 
= \gamma (z + ap - x)
\end{align*}
with initial values $p=2, q=3, w=1, x=0$, and $z=$ the desired starting value for the search minus $ap$.

We shift $z$ down by $2a$ for the following reason.
We will apply the construction of \Cref{thm:ODEtoTN} to the variables $p,q,w,z$, but \emph{not} to $x$,
because as discussed above,
we use $x$ as a model of an \emph{existing} transcription factor $X$ that influences some other transcription factor $F$ in a way we don't control, such that $F$'s concentration at any time is $f(x)$.
But for $x$ to obey the constraints in our model of transcriptional networks,
its ODE must have a single negative term $-\gamma x$.
By shifting $z$ down by $2a$,
we change the reference to it in $x$'s ODE to $z+2a$,
which cancels the term $-2a$ due to the upward shift of $p$,
restoring that $x$ has a single negative term $-\gamma x$ as required.

Note also that we have written $f(x)$,
even though $f(x)$ may not be a Laurent polynomial as required by the definition of transcriptional networks.
What we are modeling is that the environment rapidly adjusts the concentration $f(x)$ of transcription factor $F$ based on the concentration $x$ of $X$.
Although we have no control over the production of $F$ itself,
we can use it to regulate other transcription factors.
This is done using our compiler by placing a placeholder symbol
\texttt{f\_placeholder}, applying the construction to the resulting ODEs while instructing the compiler to ignore \texttt{f\_placeholder} (i.e., do not make top and bottom versions of it, and do not put any ODE for it in the transcriptional network ODEs),
and then substitute the actual (non-polynomial) expression for 
$f$ afterwards before simulating the ODEs.
See the example notebook in the code repository for details~\cite{ode2tn}.

\Cref{fig:extremum-seek-local-max-plot} shows the transcriptional network starting with various initial values of $z$, showing how $x$ changes over time.
In each case $x$ immediately shoots up to $z$, before converging (with some oscillation) to one of the two local maxima.



\end{toappendix}

\section{Conclusion}
\label{sec:conclusion}

A central message of this paper concerns the role of degradation.
Nonlinear protein degradation can be a source of effective cooperativity, enabling behaviors such as bistability and oscillation~\cite{buchler2005nonlinear};
however, we have shown that no such degradation nonlinearity is \emph{required} for analog completeness---common first-order decay suffices when transcriptional production is programmable as a positive Laurent polynomial.

We chose polynomial ODEs
as an implementation target for the reasons explained in the introduction:
they are ubiquitous, well-understood, and equivalent to other natural analog computing models such as GPAC~\cite{shannon1941mathematical} as well as being computationally powerful~\cite{bournez2017odes}.
However, since we allowed the definition of transcriptional networks to have arbitrary Laurent polynomials (a strict superset of polynomials) as production rates,
we could also have simulated arbitrary systems of ODEs with Laurent polynomials for each derivative.
Note that we could handle even more functions on the right-hand side of the ODEs (any non-hypertranscendental functions) by first converting to a system of polynomial ODEs~\cite{pour1974abstract}.
However, the construction of \Cref{alg:construction} works \emph{directly} as shown
(without any initial pre-transformation of the ODEs)
on Laurent polynomials,
hence could result in a simpler final transcriptional network than one obtained by first converting to polynomial ODEs.

Going the other direction,
we can ask how restricted production rates can be while maintaining their computational power.
As discussed in \Cref{sec:justification-laurent}, Laurent polynomials can approximate Hill functions, but inexactly.
In particular, the approximation carries a speed-accuracy tradeoff in the sense that the approximation works in the limit of some Hill dissociation constants being large or small, which translates into needing to speed up the maximal production rate to ``make room'' for such extreme constants.
It would be interesting to explore the details of this tradeoff, or other ways of approximating more commonly-used production rates with Laurent polynomials.
It would also be interesting to explore the computational power of Hill functions more directly.


\subsection{Discussion on implementation}

Although our results establish theoretical guarantees, real-world implementation of these transcriptional networks would face several practical challenges. Many of the simplifying assumptions made in our framework may differ from constraints of experimental settings.
In this section, we discuss these differences in more detail and highlight several open questions that arise when considering potential implementations.



One assumption made in \Cref{sec:justification-laurent} is that all activators of a gene must be simultaneously present and bound to the promoter to initiate transcription.
In reality, activators may bind independently, contributing additively to transcriptional activation, or they may bind cooperatively, producing nonlinear or synergistic effects that cannot be inferred from individual binding probabilities alone~\cite{chavez2015highly, zhang2015crispr}.

Another assumption we made is that the binding strengths of activators and repressors are significantly different. In reality, this is not always the case. For example, when using CRISPRa/i systems, such fine-tuned control over binding strength is not easily achievable. While certain techniques, such as introducing mismatches in the guide RNA, can partially adjust binding affinity, this tuning is not possible when a transcription factor acts as both an activator and a repressor for different genes.

Our construction assumes all transcription factors degrade at the same rate. While appropriate for degradation due to dilution,
proteases are known to degrade different proteins at different rates~\cite{bachmair1986vivo}, due for instance to the ``N-terminal residue'', the first amino acid in the peptide chain.
However, since CRISPR is one promising implementation route and all such transcription factors are actually made of the same protein (Cas9 or a variant),
differing only in the guide RNA strand,
such systems might in fact have roughly equal degradation rates even due to proteases.
Also, it would be interesting to consider systems in which different transcription factors $X_i$ have different degradation rates $-\gamma_i x_i$,
both under the model that we assume the engineer can choose degradation rates,
and under the model that degradation rates are chosen adversarially.
More generally, biological degradation can be enzyme-mediated and nonlinear rather than simple first-order decay~\cite{buchler2005nonlinear};
while our main result shows such richness is not \emph{required} for analog completeness,
characterizing what additional advantages (e.g., smaller circuits, faster convergence, or robustness to noise) it might confer is an interesting open question.


These observations point to several open directions. We note that our assumption that the presence of any repressor among multiple possible repressors is sufficient to prevent transcription is reasonable~\cite{alon2019introduction, bintu2005transcriptional}.
Based on this, we can consider different models for studying transcriptional networks:
\begin{enumerate}
    \item \textbf{Repressor-only networks:} transcriptional regulation is achieved exclusively through repressors, and gene expression occurs only in the absence of any bound repressor.
    
    \item \textbf{Limited-activation networks:} each gene promoter may be regulated by multiple repressors but by at most one activator, introducing a simple layer of transcriptional activation 
    while retaining the dominance of repression.
    
    \item \textbf{Additive-regulation networks:} multiple activators and repressors can act on the same promoter, and activators' combined influence on transcription is assumed to be additive, without cooperative or synergistic effects.
    
    \item \textbf{Single activation and repression networks:}
    Each gene promoter can be regulated by at most one activator and at most one repressor.
    
    \item \textbf{Single transcription factor networks:} Each gene promoter can be regulated by at most one transcription factor, either a repressor or an activator.
\end{enumerate}

\opt{sub,final}{\begin{credits}

\subsubsection{\ackname}We thank Ophelia S Venturelli for our earlier collaboration on related approaches to analog computation with transcriptional networks, which provided important context for the present work.
DD and ML were supported by NSF awards 2211793, 1900931 and DoE award DE-SC0024467.
DS was supported by NSF awards 2200290, SemiSynBio III: GOALI, DoE award DE-SC0024467, and a Schmidt Sciences Polymath Award.

\subsubsection{\discintname}
The authors have no competing interests to declare that are
relevant to the content of this article. 
\end{credits}}

\bibliographystyle{abbrv}
\bibliography{ref}


\end{document}